\documentclass[prl,superscriptaddress,notitlepage,twocolumn]{revtex4-1}
\usepackage[utf8]{inputenc}
\usepackage[caption=false]{subfig}
\usepackage{algorithm,algorithmic,amsfonts,amsmath,amssymb,mathtools,mathrsfs,bbm,float}
\usepackage{graphicx,epic,eepic,epsfig,latexsym,verbatim,color}
\usepackage[shortlabels]{enumitem}
 
\usepackage{tikz}
\usetikzlibrary{chains}
\usetikzlibrary{fit}

\usepackage{epsfig}
\usetikzlibrary{shapes.symbols,patterns} % for source symbols

\usepackage[strict]{changepage}
\usepackage{hyperref}
\hypersetup{colorlinks=true,citecolor=blue,linkcolor=blue,filecolor=blue,urlcolor=blue,breaklinks=true}

\usepackage[marginal]{footmisc}
\usepackage{url}
\usepackage{theorem}

\newtheorem{proposition}{Proposition}
\newtheorem{lemma}[proposition]{Lemma}

\newtheorem{theorem}[proposition]{Theorem}

\def\squareforqed{\hbox{\rlap{$\sqcap$}$\sqcup$}}
\def\qed{\ifmmode\squareforqed\else{\unskip\nobreak\hfil
\penalty50\hskip1em\null\nobreak\hfil\squareforqed
\parfillskip=0pt\finalhyphendemerits=0\endgraf}\fi}
\def\endenv{\ifmmode\;\else{\unskip\nobreak\hfil
\penalty50\hskip1em\null\nobreak\hfil\;
\parfillskip=0pt\finalhyphendemerits=0\endgraf}\fi}
\newenvironment{proof}{\noindent \textbf{{Proof~} }}{\hfill $\blacksquare$}

\newcounter{remark}
\newenvironment{remark}[1][]{\refstepcounter{remark}\par\medskip\noindent%
\textbf{Remark~\theremark #1} }{\medskip}

\newcounter{example}

\mathchardef\ordinarycolon\mathcode`\:
\mathcode`\:=\string"8000
\def\vcentcolon{\mathrel{\mathop\ordinarycolon}}
\begingroup \catcode`\:=\active
  \lowercase{\endgroup
  \let :\vcentcolon
  }

\usepackage{cleveref}
\usepackage{graphicx}
\usepackage{xcolor}

\definecolor{darkblue}{RGB}{0,76,156}
\definecolor{darkkblue}{RGB}{0,0,153}
\definecolor{blue2}{RGB}{102,178,255}
\definecolor{darkred}{RGB}{195,0,0}

\RequirePackage[framemethod=default]{mdframed}
\newmdenv[skipabove=7pt,
skipbelow=7pt,
backgroundcolor=darkblue!15,
innerleftmargin=5pt,
innerrightmargin=5pt,
innertopmargin=5pt,
leftmargin=0cm,
rightmargin=0cm,
innerbottommargin=5pt,
linewidth=1pt]{tBox}

\newmdenv[skipabove=7pt,
skipbelow=7pt,
backgroundcolor=blue2!25,
innerleftmargin=5pt,
innerrightmargin=5pt,
innertopmargin=5pt,
leftmargin=0cm,
rightmargin=0cm,
innerbottommargin=5pt,
linewidth=1pt]{dBox}

\newmdenv[skipabove=7pt,
skipbelow=7pt,
backgroundcolor=darkred!15,
innerleftmargin=5pt,
innerrightmargin=5pt,
innertopmargin=5pt,
leftmargin=0cm,
rightmargin=0cm,
innerbottommargin=5pt,
linewidth=1pt]{rBox}

\RequirePackage[framemethod=default]{mdframed}
\newmdenv[skipabove=7pt,
skipbelow=7pt,
innerleftmargin=5pt,
innerrightmargin=5pt,
innertopmargin=5pt,
leftmargin=0cm,
rightmargin=0cm,
innerbottommargin=5pt,
linewidth=1pt]{cBox}

\newcommand{\nc}{\newcommand}
\nc{\bra}[1]{\langle#1|}
\nc{\ket}[1]{|#1\rangle}
\nc{\ketbra}[2]{|#1\rangle\!\langle#2|}
\nc{\braket}[2]{\langle#1|#2\rangle}

\nc{\proj}[1]{| #1\rangle\!\langle #1 |}
\nc{\avg}[1]{\langle#1\rangle}
\nc{\rank}{\operatorname{Rank}}
\nc{\smfrac}[2]{\mbox{$\frac{#1}{#2}$}}
\nc{\tr}{{\sf Tr}}
\nc{\ox}{\otimes}
\nc{\dg}{\dagger}
\nc{\dn}{\downarrow}
\nc{\cA}{{\cal A}}
\nc{\cB}{{\cal B}}
\nc{\cC}{{\cal C}}
\nc{\cD}{{\cal D}}
\nc{\cE}{{\cal E}}
\nc{\cF}{{\cal F}}
\nc{\cG}{{\cal G}}
\nc{\cH}{{\cal H}}
\nc{\cI}{{\cal I}}
\nc{\cJ}{{\cal J}}
\nc{\cK}{{\cal K}}
\nc{\cL}{{\cal L}}
\nc{\cM}{{\cal M}}
\nc{\cN}{{\cal N}}
\nc{\cO}{{\cal O}}
\nc{\cP}{{\cal P}}
\nc{\cQ}{{\cal Q}}
\nc{\cR}{{\cal R}}
\nc{\cS}{{\cal S}}
\nc{\cT}{{\cal T}}
\nc{\cU}{{\cal U}}
\nc{\cV}{{\cal V}}
\nc{\cX}{{\cal X}}
\nc{\cY}{{\cal Y}}
\nc{\cZ}{{\cal Z}}
\nc{\cW}{{\cal W}}
\nc{\csupp}{{\operatorname{csupp}}}
\nc{\qsupp}{{\operatorname{qsupp}}}
\nc{\var}{{\operatorname{var}}}
\nc{\rar}{\rightarrow}
\nc{\lrar}{\longrightarrow}
\nc{\polylog}{{\operatorname{polylog}}}
\nc{\wt}{{\operatorname{wt}}}
\nc{\supp}{{\operatorname{supp}}}

\nc{\argmin}{{\operatorname{argmin}}}

\def\x{\xi}

\nc{\RR}{{{\mathbb R}}}
\nc{\CC}{{{\mathbb C}}}
\nc{\FF}{{{\mathbb F}}}
\nc{\NN}{{{\mathbb N}}}
\nc{\ZZ}{{{\mathbb Z}}}
\nc{\PP}{{{\mathbb P}}}
\nc{\QQ}{{{\mathbb Q}}}
\nc{\UU}{{{\mathbb U}}}
\nc{\EE}{{{\mathbb E}}}
\nc{\id}{{\operatorname{id}}}

\nc{\CHSH}{{\operatorname{CHSH}}}

\nc{\rU}{\mbox{U}}

\nc{\ob}[1]{#1}

\nc{\SEP}{{\text{\rm SEP}}}
\nc{\NS}{{\text{\rm NS}}}
\nc{\LOCC}{{\text{\rm LOCC}}}
\nc{\PPT}{{\text{\rm PPT}}}
\nc{\EXT}{{\text{\rm EXT}}}
\nc{\Sym}{{\operatorname{Sym}}}

\nc{\ERLO}{{E_{\text{r,LO}}}}
\nc{\ERLOCC}{{E_{\text{r,LOCC}}}}
\nc{\ERPPT}{{E_{\text{r,PPT}}}}
\nc{\ERLOCCinfty}{{E^{\infty}_{\text{r,LOCC}}}}
\nc{\Aram}{{\operatorname{\sf A}}}

\newcommand{\Choi}{Choi-Jamio\l{}kowski }

\usepackage{hyperref}
\hypersetup{colorlinks=true,citecolor=blue,linkcolor=blue,filecolor=blue,urlcolor=blue,breaklinks=true}

\makeatletter
\def\grd@save@target#1{%
  \def\grd@target{#1}}
\def\grd@save@start#1{%
  \def\grd@start{#1}}
\tikzset{
  grid with coordinates/.style={
    to path={%
      \pgfextra{%
        \edef\grd@@target{(\tikztotarget)}%
        \tikz@scan@one@point\grd@save@target\grd@@target\relax
        \edef\grd@@start{(\tikztostart)}%
        \tikz@scan@one@point\grd@save@start\grd@@start\relax
        \draw[minor help lines,magenta] (\tikztostart) grid (\tikztotarget);
        \draw[major help lines] (\tikztostart) grid (\tikztotarget);
        \grd@start
        \pgfmathsetmacro{\grd@xa}{\the\pgf@x/1cm}
        \pgfmathsetmacro{\grd@ya}{\the\pgf@y/1cm}
        \grd@target
        \pgfmathsetmacro{\grd@xb}{\the\pgf@x/1cm}
        \pgfmathsetmacro{\grd@yb}{\the\pgf@y/1cm}
        \pgfmathsetmacro{\grd@xc}{\grd@xa + \pgfkeysvalueof{/tikz/grid with coordinates/major step}}
        \pgfmathsetmacro{\grd@yc}{\grd@ya + \pgfkeysvalueof{/tikz/grid with coordinates/major step}}
        \foreach \x in {\grd@xa,\grd@xc,...,\grd@xb}
        \node[anchor=north] at (\x,\grd@ya) {\pgfmathprintnumber{\x}};
        \foreach \y in {\grd@ya,\grd@yc,...,\grd@yb}
        \node[anchor=east] at (\grd@xa,\y) {\pgfmathprintnumber{\y}};
      }
    }
  },
  minor help lines/.style={
    help lines,
    step=\pgfkeysvalueof{/tikz/grid with coordinates/minor step}
  },
  major help lines/.style={
    help lines,
    line width=\pgfkeysvalueof{/tikz/grid with coordinates/major line width},
    step=\pgfkeysvalueof{/tikz/grid with coordinates/major step}
  },
  grid with coordinates/.cd,
  minor step/.initial=.2,
  major step/.initial=1,
  major line width/.initial=2pt,
}
\makeatother

\usepackage{thmtools}
\usepackage{thm-restate}
\usepackage{etoolbox}
\makeatletter
\def\problem@s{}
\newcounter{problems@cnt}

\newcommand{\allproblems}{\problem@s}
\makeatother
\newcommand{\idop}{\mathbbm{1}} %identity operator
\newcommand{\lrp}[1]{\left( #1 \right)}
\newcommand{\lrb}[1]{\left[ #1 \right]}
\newcommand{\lrc}[1]{\left\{ #1 \right\}}
\newcommand{\lrv}[1]{\left| #1 \right|}
\newcommand{\lrV}[1]{\left\| #1 \right\|}
\begin{document}
\title{Shadow Simulation of Quantum Processes}
\author{Xuanqiang Zhao}
\email{xqzhao7@connect.hku.hk}
\affiliation{QICI Quantum Information and Computation Initiative, Department of Computer Science, The University of Hong Kong, Pokfulam Road, Hong Kong}

\author{Xin Wang}
\email{felixxinwang@hkust-gz.edu.cn}
\affiliation{Thrust of Artificial Intelligence, Information Hub,\\
The Hong Kong University of Science and Technology (Guangzhou), Guangdong 511453, China}

\author{Giulio Chiribella}
\email{giulio@cs.hku.hk}
\affiliation{QICI Quantum Information and Computation Initiative, Department of Computer Science, The University of Hong Kong, Pokfulam Road, Hong Kong}
\affiliation{Quantum Group, Department of Computer Science, University of Oxford, Wolfson Building, Parks Road, Oxford, OX1 3QD, United Kingdom}
\affiliation{Perimeter Institute for Theoretical Physics, 31 Caroline Street North, Waterloo, Ontario, Canada}

%%%%%%%%%%%%%%%%%%%%%%%%%%%%%%%%%%%%%%%%%
%%%%%%%%%%%%%%%%%%%%%%%%%%%%%%%%%%%%%%%%%
\begin{abstract}
We introduce the task of shadow process simulation, where the goal is to simulate the estimation of the expectation values of arbitrary quantum observables at the output of a target physical process.
When the sender and receiver share random bits or other no-signaling resources, we show that the performance of shadow process simulation exceeds that of conventional process simulation protocols in a variety of scenarios including communication, noise simulation, and data compression.
Remarkably, we find that there exist scenarios where shadow simulation provides increased statistical accuracy without any increase in the number of required samples.
\end{abstract}

\date{\today}
\maketitle

%%%%%%%%%%%%%%%%%%%%%%%%%%%%%%%%%%%%%%%%%
%%%%%%%%%%%%%%%%%%%%%%%%%%%%%%%%%%%%%%%%%
\textcolor{black}{\textbf{\emph{Introduction.---}}}
``What is a quantum state?'' is one of the central questions in the foundations of quantum mechanics. A minimal interpretation is that a quantum state is a compact way to represent the expectation values of all possible observables associated to a given system. This interpretation may suggest that transmitting a quantum state from a place to another is equivalent to transferring information about the expectation values of arbitrary observables.
In this Letter, we show that this equivalence does not hold when the sender and receiver share random bits or more general no-signaling resources~\cite{leung2015power,duan2015no,fang2019quantum}: in general, transferring information about all possible expectation values is a much less demanding task than transferring the quantum state itself. Some instances of this phenomenon can be derived from existing results on error mitigation~\cite{temme2017error, li2017efficient, endo2018practical}, while other more radical instances emerge from a new task that we name {\em shadow simulation of quantum processes.}

The settings of shadow simulation are illustrated in Fig.~\ref{fig:simulation_setting}.
The goal is to simulate the estimation of the expectation values at the output of a quantum communication channel $\cM$ (with input system $A'$ and output system $B'$), by utilizing  another  quantum channel $\cN$ (with input system $A$ and output system $B$) in place of channel $\cM$. Here, $A$, $A'$, $B$, and $B'$ are arbitrary finite-dimensional quantum systems, possibly of different dimensions.
A sender, Alice, has access to system $A'$, which may be entangled with a reference system $R$ in the laboratory of a receiver, Bob.
Systems $A'$ and $R$ are initially in a quantum state $\rho$, possibly unknown to Alice.
On his side, Bob has a device that measures an  observable $O$ (possibly unknown to him) on systems $B'$ and $R$.
The goal is to estimate the expectation value $\tr[O (\cM \ox \cI_R)(\rho)]$, where $\cI_R$ is the identity channel on system $R$. To achieve this goal, Alice and Bob perform encoding and decoding operations $\cE_j$ and $\cD_j$, respectively, coordinating their actions using shared random bits. To estimate the expectation value $\tr[O (\cM\otimes \cI_R)(\rho)]$, Bob measures the observable $O$ on systems $B'$ and $R$ and then postprocesses the measurement outcomes. The protocol is successful if Bob's estimate deviates from the true value $\tr[O (\cM\otimes \cI_R)(\rho)]$ by less than a given error tolerance $\epsilon$, for every possible state $\rho$, for every possible observable $O$, and for every possible reference system $R$.

\begin{figure}[t]
    \centering
    \includegraphics[width=\columnwidth]{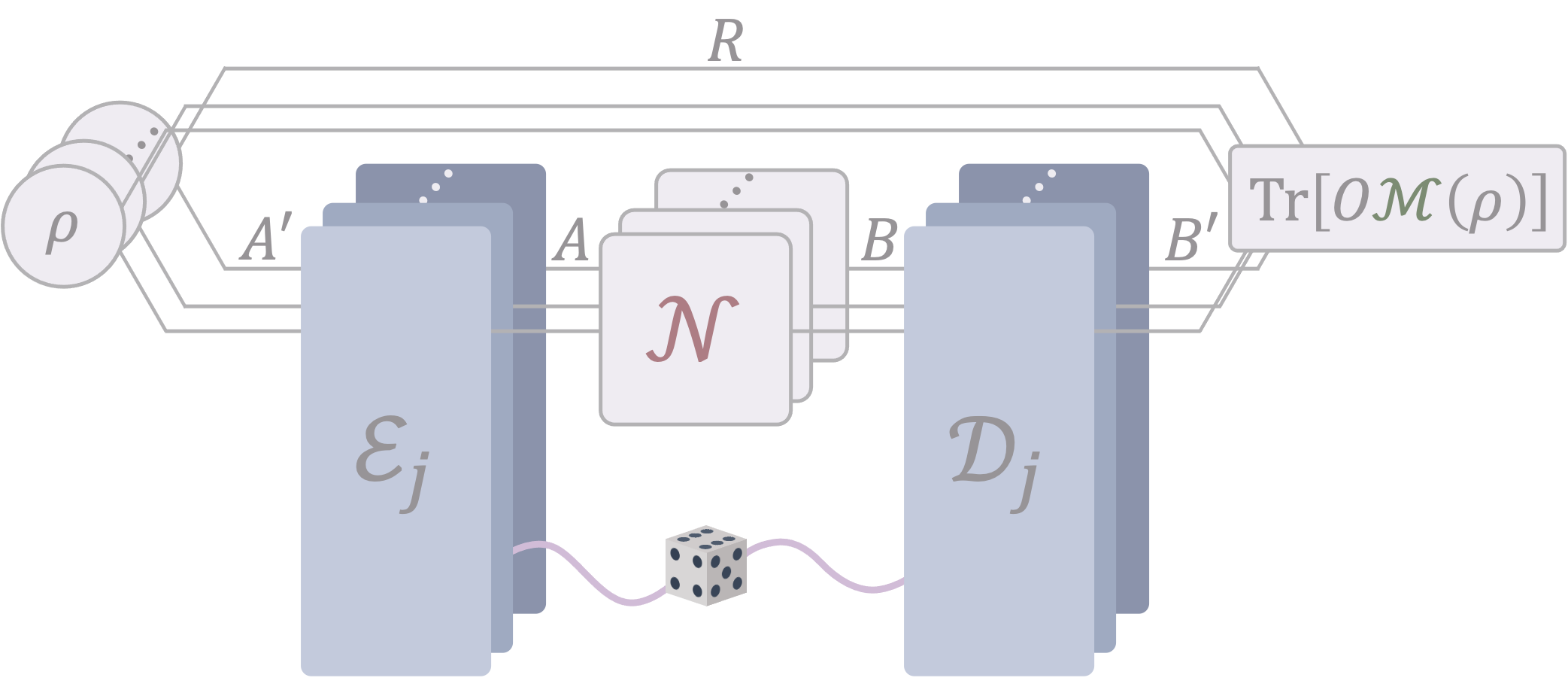}
    \caption{{\bf Shadow simulation of quantum channels.}  A sender (Alice) and a receiver (Bob) are connected through a quantum communication channel $\cN$. A shadow simulation protocol  enables Bob to estimate the expectation value $\tr[ O (\cM \otimes \cI_R)(\rho)]$ of an arbitrary observable $O$ on the output state $(\cM\otimes \cI_R) (\rho)$ produced by a target channel $\cM$ when acting on an arbitrary input state $\rho$ of a system $A'$ in Alice's laboratory, possibly entangled with a reference system $R$ in Bob's laboratory. To achieve this task, Alice and Bob can coordinate their operations $\cE_j$ and $\cD_j$ by sharing random bits, or other no-signaling resources.
    }
    \label{fig:simulation_setting}
\end{figure}

Shadow simulation is  a generalization of the task of quantum channel simulation~\cite{kretschmann2004tema, berta2011quantum,bennett2014quantum,duan2015no,leung2015power,Wang2016g,fang2019quantum}, with the crucial difference that shadow simulation does not aim at reproducing the channel's output states, but only their ``shadow information''~\cite{zhao2022information, aaronson2018shadow}, {\em i.e.}, the information about the expectation values of all possible observables. 
This generalization is useful for quantum information processing in the noisy intermediate-scale quantum (NISQ) era~\cite{preskill2018quantum}, where classical postprocessing of expectation values can be used to
simulate larger quantum memories~\cite{bravyi2016trading, peng2020simulating, mitarai2021overhead, piveteau2022circuit}, probe properties of quantum systems~\cite{buscemi2013direct, huang2020predicting, elben2023randomized}, mitigate errors~\cite{li2017efficient, temme2017error, piveteau2021error, lostaglio2021error, suzuki2022quantum, cai2022quantum, liu2024virtual}, simulate unphysical operations~\cite{jiang2021physical, regula2021operational, zhao2023power}, and distill resources~\cite{yuan2023virtual,yuan2024virtual}.

Remarkably, we find that shadow simulation is not constrained by the limits of conventional channel simulation.
For example, we show that one can perfectly shadow simulate the transmission of an arbitrarily large number of qubits using only a single-qubit channel,
a task that cannot be achieved in any conventional channel simulation protocol. In this example, the advantage over conventional protocols is achieved at the price of an increase in the number of samples needed to accurately estimate the expectation values.
Quite surprisingly, we also find that there also  exist scenarios where shadow simulation achieves lower error than conventional channel simulation without any sampling overhead, provided that Alice and Bob have access to no-signaling quantum resources.
Overall, our results show that shared classical randomness and classical postprocessing are valuable resources for quantum communication and other quantum technologies.

%%%%%%%%%%%%%%%%%%%%%%%%%%%%%%%%%%%%%%%%%
%%%%%%%%%%%%%%%%%%%%%%%%%%%%%%%%%%%%%%%%%
\textcolor{black}{\textbf{\emph{Framework: virtual supermaps.---}}}
In conventional quantum channel simulation, the aim is to simulate the action of a target channel $\cM$ (with input $A'$ and output $B'$) on a generic input state by using an available channel $\cN$ (with input $A$ and output $B$). To convert channel $\cN$ into an approximation of channel $\cM$, Alice and Bob perform encoding and decoding channels $\cE$ (with input $A'$ and output $A$) and $\cD$ (with input $B$ and output $B'$), respectively, thus obtaining the new channel $ \cD  \circ \cN \circ \cE$.
In shadow simulation, instead, Bob can sample his decoding operation  from a set of channels $\{\cD_j\}$ and, after the output has been measured, he can postprocess the measurement statistics, taking arbitrary linear combinations of the expectation values obtained for different values of $j$. In this way, Bob can reproduce the action of a linear map of the form $\widetilde \cD =  \sum_j  \,  \lambda_j\,  \cD_j$, where $\{\lambda_j\}$ are arbitrary real numbers~\cite{buscemi2013direct, temme2017error, jiang2021physical, piveteau2022quasiprobability}.
Note that, in general, $\widetilde \cD$ is not a valid quantum channel (completely positive trace-preserving map), but rather a {\em virtual channel},  {\em i.e.} a linear map that is Hermitian-preserving ({\em i.e.}, it maps Hermitian operators into Hermitian operators) and trace-scaling ({\em i.e.}, it is proportional to a trace-preserving map, with the proportionality constant---hereafter called the {\em scaling factor}---being an arbitrary real number)~\cite{zhao2022information, yuan2023virtual, parzygnat2023virtual, yao2023optimal, chen2023a}.

More generally, Bob may share random bits  with Alice and coordinate his local operations with hers, as in Fig.~\ref{fig:simulation_setting}. Hence, Bob's postprocessing gives rise to linear maps of the form
\begin{align}\label{supermap}
\widetilde\cS (\cN) = \sum_{j} \lambda_j \cD_j \circ \cN \circ \cE_j.
\end{align}
Mathematically, the linear map $\widetilde\cS$ is an example of a supermap~\cite{chiribella2008transforming, chiribella2009theoretical, chiribella2013quantum},
that is, a map acting on the vector space spanned by quantum operations. Unlike most supermaps considered so far, however, $\widetilde \cS$ generally does not transform quantum channels into quantum channels, due to the possible presence of negative coefficients in the set $\{\lambda_j\}$.
Instead, $\widetilde \cS$ transforms Hermitian-preserving maps into Hermitian-preserving maps, and, for every $\lambda \in \mathbb R$, trace-scaling maps with scaling factor $\lambda$ into trace-scaling maps with scaling factor $  \lambda_* \,   \lambda$, for some real number $\lambda_*$  (in the case of Eq.~\eqref{supermap},  $\lambda_*  = \sum_j \lambda_j $). We call the supermaps with these two properties {\em virtual supermaps}.

Virtual supermaps $\widetilde \cS$ mapping virtual channels from $A$ to $B$ into virtual channels from $A'$ to $B'$ are into one-to-one correspondence with virtual bipartite channels $\widehat \cS$ transforming Hermitian operators on the composite system $A'B$ into Hermitian operators on  the composite system $AB'$~\cite{chiribella2008transforming, chiribella2008quantum, gour2019comparison}. The correspondence can be made explicit by decomposing the action of the virtual supermap $\widetilde \cS$ as $\widetilde \cS (\cN) = \sum_k \, \cB_k \circ \cN \circ \cA_k$, where $\cA_k$ and $\cB_k$ are suitable linear maps from $A'$ to $A$ and from $B$ to $B'$, respectively. Using this decomposition, one can define the map $\widehat \cS \coloneqq \sum_k \cA_k \otimes \cB_k$, which is guaranteed to be a virtual channel (see Appendix~A~\cite{supplemental}).

A virtual supermap $\widetilde \cS$ of the special form~\eqref{supermap} will be called a \emph{(randomness-assisted) shadow simulation code}. For a shadow simulation code $\widetilde \cS$, the corresponding virtual bipartite channel is $\widehat \cS \coloneqq \sum_j \, \lambda_j\, \cE_j\otimes \cD_j$.
It is rather straightforward to see that this virtual channel is {\em no-signaling}, meaning that for every pair of operators $\rho_{A'}$ and $\rho_B$ acting on $A'$ and $B$, respectively, the operator $\tr_{B'} [\widehat \cS (\rho_{A'}\otimes \rho_B)]$ is independent of $\rho_B$, and the operator $\tr_{A} [\widehat \cS (\rho_{A'} \otimes \rho_{B})]$ is independent of $\rho_{A'}$ (here, $\tr_X$ denotes the partial trace over the Hilbert space associated to system $X$, for $X=  A$ or $X=B'$).
The converse is less straightforward but still true. In short, we have the following characterization:

\begin{theorem}\label{thm:ns_code}
    A virtual supermap $\widetilde \cS$ is a randomness-assisted shadow simulation code if and only if the corresponding virtual channel $\widehat \cS$ is no-signaling.
\end{theorem}

The proof is provided in Appendix~A~\cite{supplemental}. An important consequence of Theorem \ref{thm:ns_code} is that shared randomness and classical postprocessing can be used to simulate arbitrary no-signaling resources. Explicitly, a no-signaling resource is represented by a quantum no-signaling channel $\widehat \cS$~\cite{piani2006properties, chiribella2012perfect} with input $A'B$ and output $AB'$. Using the no-signaling channel $\widehat \cS$, Alice and Bob can implement the corresponding supermap $\widetilde \cS$. In turn, Theorem~\ref{thm:ns_code} guarantees that the supermap $\widetilde \cS$ can be implemented using only shared randomness and classical postprocessing. The same conclusion applies to protocols using local operations and shared entanglement, which are a special case of no-signaling resources.

%%%%%%%%%%%%%%%%%%%%%%%%%%%%%%%%%%%%%%%%%
%%%%%%%%%%%%%%%%%%%%%%%%%%%%%%%%%%%%%%%%%
\textcolor{black}{\textbf{\emph{No-signaling assisted shadow simulation.---}}}
Quantum no-signaling resources have been extensively studied in quantum Shannon theory~\cite{leung2015power,duan2015no,fang2019quantum}.
In conventional channel simulation, Alice and Bob have the assistance of a fixed no-signaling channel $\widehat \cS$. In shadow simulation, instead, we allow them to sample over a set of no-signaling channels $\{ \widehat \cS_j \}$. Using classical postprocessing, they can then reproduce the virtual supermap $\widetilde \cS \coloneqq \sum_j \,\lambda_j\, \widetilde \cS_j$, where $\{\lambda_j\}$ are arbitrary real coefficients. We call a supermap of this type a {\em no-signaling shadow simulation code}.

The randomization over different settings generally comes at the price of an increased sampling cost, meaning that more rounds of data collection are needed to estimate the desired expectation values~\cite{jiang2021physical, regula2021operational, zhao2023power}. For example, one way to realize a virtual supermap $\widetilde \cS = \sum_j \,\lambda_j\, \widetilde \cS_j$ is to sample the quantum supermap $\cS_j$ with probability $p_j \coloneqq |\lambda_j|/\lambda_{\rm tot}$, $\lambda_{\rm tot}  \coloneqq  \sum_i  |\lambda_i|$  and, in the data processing stage, to multiply the estimated expectation value of $O$  by ${\rm sign}  (\lambda_j) \, \lambda_{\rm tot}$. This multiplication, however, generally increases the root mean square error (RMSE) by a factor $\lambda_{\rm tot}$: to estimate the desired quantity with RMSE lower than $\epsilon$, one needs to have measurement data with RMSE lower than $\epsilon/\lambda_{\rm tot}$, which generally requires increasing the number of samples by a factor $\lambda_{\rm tot}^2$ (see Appendix~A~\cite{supplemental} for more detail).
Optimizing the decomposition of $\widetilde \cS$ over all possible linear combinations of quantum supermaps we then obtain the following definition of the {\em sampling cost}:
\begin{align*}
    c_{\rm smp}\lrp{\widetilde{\cS}} &\coloneqq \inf \Big\{p_+ + p_- ~\Big|~ \widehat{\cS} = p_+ \widehat{\cS}^+ - p_- \widehat{\cS}^-, \nonumber\\
    &\qquad \qquad \qquad  p_\pm \in \mathbb{R}^+,~ \widehat{\cS}^\pm \in {\rm CPTP}\cap{\rm NS}\Big\},
\end{align*}
where ${\rm CPTP}$ and ${\rm NS}$ denote the set of completely positive, trace-preserving maps and the set of no-signaling virtual processes, respectively, and $\mathbb{R}^+$ is the set of non-negative real numbers.
Note that every conventional channel simulation protocol has $c_{\rm smp}\lrp{\widetilde{\cS}} = 1$, since the map $\widehat{\cS}$ is already a no-signaling channel.

An operational measure of error, both in the conventional channel simulation scenario and in shadow simulation, is the {\em simulation error}, defined  as ${\rm err}(\widetilde{\cS}; \cN, \cM) \coloneqq \| \widetilde \cS  (\cN) - \cM \|_{\diamond}/2$, where $\|  \cdot \|_\diamond$ is the diamond-norm distance~\cite{kitaev1997quantum}. Operationally, ${\rm err}(\widetilde{\cS}; \cN, \cM)$ can be interpreted as the worst case error in  measuring (normalized) observables involving a reference system (see Appendix~A~\cite{supplemental} for more detail).
The optimal performance in  the task of shadow simulation can be quantified by two parameters. One is the minimum error achievable with a sampling cost bounded by $\gamma$:
\begin{align*}
    \varepsilon^*_{\gamma, \rm NS} (\cN, \cM) \coloneqq \inf_{\widehat{\cS} \in {\rm NS}} \lrc{{\rm err}\lrp{\widetilde{\cS}; \cN, \cM} ~\middle|~ c_{\rm smp}\lrp{\widetilde{\cS}} \leq \gamma} \,.
\end{align*}
The other is the minimum sampling cost needed to guarantee an error  below an error tolerance $\varepsilon$:
\begin{align*}
    \gamma^*_{\varepsilon, \rm NS} \lrp{\cN, \cM} \coloneqq \inf_{\widehat{\cS} \in {\rm NS}} \lrc{c_{\rm smp}\lrp{\widetilde{\cS}} ~\middle|~ {\rm err}\lrp{\widetilde{\cS}; \cN, \cM} \leq \varepsilon}.
\end{align*}
Both quantities can be computed efficiently by semidefinite programs (SDPs) given in Appendix~B~\cite{supplemental}. In the following, we illustrate the power of shadow simulation in three applications.

%%%%%%%%%%%%%%%%%%%%%%%%%%%%%%%%%%%%%%%%%
%%%%%%%%%%%%%%%%%%%%%%%%%%%%%%%%%%%%%%%%%
\textcolor{black}{\textbf{\emph{Zero-error shadow communication.---}}}
Quantum communication can be viewed as a special case of channel simulation: the simulation of an identity channel acting on a given number of qubits. Here we consider the zero-error scenario~\cite{shannon1956zero, cubitt2011zero, duan2015no, wang2016quantum}, corresponding to an exact simulation of the identity channel. In this case, the virtual supermap $\cS$ satisfies the condition $\widetilde{\cS} (\cN) = \cI_d$, where $\cI_d$ denotes the identity channel on a $d$-dimensional quantum system. We define the one-shot zero-error shadow capacity assisted by no-signaling resources as
\begin{align}\label{def:shadowcap}
    Q^{(1)}_{\gamma, \rm NS} (\cN) \coloneqq \sup_{d,\widehat{\cS} \in {\rm NS}}& \Big\{\log_2 d ~\Big|~ \widetilde{\cS}(\cN) = \cI_d,~ c_{\rm smp}\lrp{\widetilde{\cS}} \leq \gamma \Big\},
\end{align}
where the dimension $d$ is optimized over all positive integers. Here the term ``one-shot'' refers to the fact that the communication protocol only involves quantum operations on the inputs and outputs of a {\em single} use of channel $\cN$.
Note, however, that $\widetilde \cS$ is a virtual supermap, implemented by sampling over multiple supermaps and averaging the results across multiple rounds, each of which involves one  use of $\cN$.

In Appendix~C~\cite{supplemental}, we provide an explicit SDP expression for the shadow capacity for every $\gamma$. This expression extends the previously known expression for the one-shot zero-error quantum capacity assisted by no-signaling resources ~\cite{duan2015no, wang2016semidefinite}, which can be retrieved in the special case $\gamma=1$. For $\gamma > 1$, we show that the zero-error shadow capacity is generally larger than the zero-error quantum capacity. A concrete example is as follows: 
\begin{theorem}\label{thm:depo_capacity}
    Let $\cN_{{\rm depo}, p}(\rho) = p\rho + (1-p)\idop_2/2$ be a single-qubit depolarizing channel, where $p\in [0,1]$ is a probability and $\idop_2$ is the identity operator on $\mathbb{C}^2$.
    For $\gamma \geq 1$, the one-shot zero-error shadow capacity assisted by no-signaling resources is
    \begin{align}\label{eq:depo_shadow_capacity}
        Q^{(1)}_{\gamma, \rm NS}(\cN_{{\rm depo}, p}) = \log_2 \left\lfloor \sqrt{2p\gamma + p + 1} \right\rfloor.
    \end{align}
\end{theorem}
Equation~\eqref{eq:depo_shadow_capacity} shows that the shadow capacity can become arbitrarily large as $\gamma$ grows, provided that the channel is not completely depolarizing ($p \neq 0$). In other words, a qubit depolarizing channel can be used to transmit the expectation values of all observables on a quantum system of arbitrarily high dimension $d$, at the price of an increased sampling cost $\gamma$.

It is useful to compare the above finding with existing results about error mitigation. Error mitigation protocols, such as those in Refs.~\cite{takagi2021optimal, jiang2021physical}, can be used to transmit arbitrary expectation values {\em on a single-qubit state} through repeated uses of a single-qubit depolarizing channel. This fact is interesting because, for $p\le 1/3$, the depolarizing channel is entanglement-breaking~\cite{horodecki2003entanglement} and therefore it cannot reliably transmit quantum states, even if used infinitely many times.
Theorem~\ref{thm:depo_capacity} takes this observation to a much stronger level: not only can a qubit depolarizing channel transmit all expectation values for a single qubit, but also it can transmit the expectation values for  quantum systems of arbitrarily large dimension.

Now, recall that Theorem~\ref{thm:ns_code} guarantees that every no-signaling code can be simulated with local operations and shared randomness (possibly at the price of a larger sampling cost.)
Combining this fact with Theorem~\ref{thm:depo_capacity}, we obtain that the assistance of shared randomness gives depolarizing channels  an arbitrarily large  shadow capacity for sufficiently large values of the sampling cost. The same argument applies to shadow simulation codes assisted by shared entanglement.
Remarkably, these phenomena are not limited to the depolarizing channel, but apply in general to every quantum channel achieving a positive value of the zero-error no-signaling assisted shadow capacity~\eqref{def:shadowcap} for at least one value of $\gamma$. The proof of this fact will be provided at the end of the next section.

%%%%%%%%%%%%%%%%%%%%%%%%%%%%%%%%%%%%%%%%%
%%%%%%%%%%%%%%%%%%%%%%%%%%%%%%%%%%%%%%%%%
\textcolor{black}{\textbf{\emph{Zero-error channel formation.---}}}
The dual task to communication is the simulation of a target noisy channel using a noiseless channel. For this task, which we call {\em channel formation}, the key quantity is the simulation cost, defined as the number of noiseless qubits that must be sent from the sender to the receiver. We define the one-shot zero-error shadow simulation cost of channel $\cM$ under no-signaling resources as
\begin{align}\label{shadowsimcost}
    S^{(1)}_{\gamma, \rm NS}(\cM) \coloneqq \inf_{d,\widehat{\cS} \in {\rm NS}}& \Big\{ \log_2 d ~\Big|~ \widetilde{\cS}(\cI_d) = \cM,~ c_{\rm smp}\lrp{\widetilde{\cS}} \leq \gamma \Big\},
\end{align}
where the dimension $d$ is optimized over positive integers. This quantity generalizes the one-shot zero-error simulation cost studied in the conventional quantum channel simulation scenario~\cite{duan2015no, fang2019quantum}, which coincides with the shadow simulation cost~\eqref{shadowsimcost} for $\gamma = 1$. In Appendix~D~\cite{supplemental}, we provide an SDP for the shadow simulation cost, and we show that the simulation cost can generally be reduced by increasing the sampling cost.

Besides simulating noisy channels, shadow simulation also allows a simulation of high-dimensional noiseless channels using low-dimensional ones. This simulation can also be viewed as a form of quantum compression, where the goal is to store the expectation values of all possible observables. Alternatively, one can view this shadow simulation as the simulation of a high-dimensional quantum measurement using a low-dimensional one ~\cite{ioannou2022simulability}.

The minimum sampling cost for simulating a higher-dimensional identity channel is provided by the following theorem:
\begin{theorem}
    Given identity channels $\cI_d$ and $\cI_{d'}$ with $d' \geq d \geq 2$, the minimum sampling cost of an exact shadow simulation of $\cI_{d'}$ using $\cI_d$ and no-signaling resources is
    \label{thm:high_id_sim_cost}
    \begin{align}\label{eq:low_id_to_high_id_cost}
        \gamma^*_{0, \rm NS} \lrp{\cI_d, \cI_{d'}} = 2 \lrp{\frac{d^{\prime}}{d}}^2 - 1.
    \end{align}
\end{theorem}

Theorem~\ref{thm:high_id_sim_cost} has two important consequences. The first is that  the shadow simulation cost of {\em every} quantum channel can be lowered to $1$ qubit by allowing a sufficiently large sampling overhead. Indeed, suppose that channel $\cM$ has zero-error quantum simulation cost  $\log_2 d'$, meaning that $\cM$ can be simulated using $\cI_{d'}$ through a quantum no-signaling code. In turn, Eq.~\eqref{eq:low_id_to_high_id_cost} implies that $\cI_{d'}$ can be shadow simulated using a qubit identity channel, with a sampling cost $(d'^2 - 2)/2$. Composing the two simulations, one gets a shadow simulation of channel $\cM$ via the identity channel $\cI_2$ with sampling cost $(d'^2 - 2)/2$.

Second, Theorem~\ref{thm:high_id_sim_cost} implies that every channel $\cN$ with positive zero-error shadow capacity $Q^{(1)}_{\gamma, \rm NS}$ for some $\gamma$ can shadow simulate every other channel $\cM$. Indeed, the positive capacity condition implies that $\cN$ can perfectly shadow simulate the identity channel $\cI_d$ for some dimension $d\ge 2$, while the argument in the previous paragraph implies that the identity channel $\cI_2$  (and therefore any identity channel $\cI_d$ with $d\ge 2$) can perfectly simulate any other quantum channel $\cM$.  Hence, the concatenation of these two simulations yields a perfect simulation of channel $\cM$ using channel $\cN$.

%%%%%%%%%%%%%%%%%%%%%%%%%%%%%%%%%%%%%%%%%
%%%%%%%%%%%%%%%%%%%%%%%%%%%%%%%%%%%%%%%%%
\textcolor{black}{\textbf{\emph{Approximate shadow simulation.---}}}
In the zero-error scenario, we have seen that the shadow capacity and shadow simulation cost coincide with the conventional capacity and simulation cost for $\gamma=1$. In stark contrast, we now show that in the approximate scenario, shadow simulation can sometime achieve lower error than conventional channel simulation even for   $\gamma=1$, that is, without incurring in any sampling overhead.

We consider three paradigmatic examples of quantum channels: the single-qubit amplitude damping channel $\cN_{\rm AD}$ with two Kraus operators $\proj{0} + \sqrt{p}\proj{1}$ and $\sqrt{1-p}\ketbra{0}{1}$, the single-qubit dephasing channel $\cN_{\rm deph}(\cdot) = p(\cdot) + (1-p){\rm diag}(\cdot)$, and the single-qubit depolarizing channel $\cN_{\rm depo}(\cdot) = p(\cdot) + (1-p)\tr[\cdot]\idop_2/2$.
Figure~\ref{fig:min_err} shows the minimum error (measured by diamond distance) for fixed $p=0.9$ and variable cost budgets ranging from $0.9$ to $1.2$.

Surprisingly, shadow simulation codes achieve a smaller error at the same cost, or even at a lower cost compared with quantum simulation codes (whose sampling cost is $1$). The difference is evident in the plot of noise simulation, where the minimum error of quantum codes is almost twice of the shadow simulation codes.
This implies that classical postprocessing can enhance the transmission of expectation values even if no sampling overhead is involved.
Technically, the smaller error arises from the fact that the virtual channels achievable by unit-cost shadow simulation protocols can be closer  to the target channel than all quantum channels achievable by conventional quantum channel simulation.

\begin{figure}[t]
    \centering
    \includegraphics[width=0.49\textwidth]{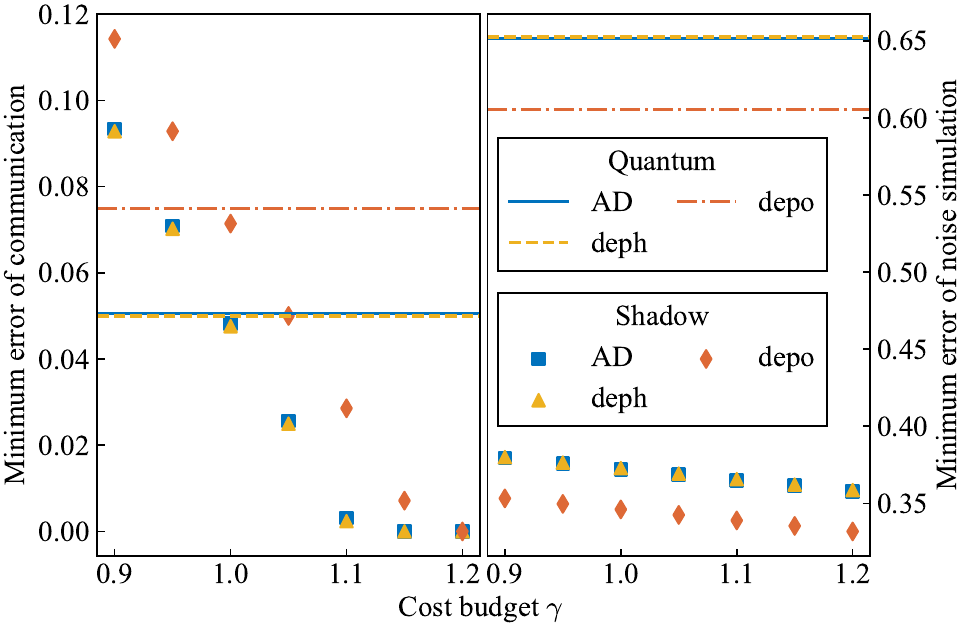}
    \caption{{\bf Approximate shadow communication and shadow simulation.} The two rectangles on the left and right sides of the figure provide the minimum  error  for the tasks of shadow communication (left)  and shadow simulation  (right), respectively. The error is plotted as a function of the cost budget $\gamma$ for three different quantum channels: the amplitude damping channel (AD), the depolarizing channel (depo), and the dephasing channel (deph). For shadow communication, the goal is to simulate a qubit identity channel using a qubit amplitude damping-depolarizing-dephasing channel with $p=0.9$. For shadow simulation, the goal is to simulate two independent uses of qubit amplitude damping-depolarizing-dephasing channel with $p=0.9$.
    The minimum error of conventional quantum channel simulation is indicated by horizontal lines, with the blue solid lines, red dash-dotted lines, and yellow dashed lines corresponding to amplitude damping, depolarizing, and dephasing channels, respectively.
    }
    \label{fig:min_err}
\end{figure}

%%%%%%%%%%%%%%%%%%%%%%%%%%%%%%%%%%%%%%%%%
%%%%%%%%%%%%%%%%%%%%%%%%%%%%%%%%%%%%%%%%%
\textcolor{black}{\textbf{\emph{Conclusions.---}}}
In this Letter, we introduced the task of shadow simulation of quantum channels, showing that transmitting and processing expectation values of arbitrary observables is generally less demanding than transmitting and processing quantum states. Besides their foundational interest, our results are relevant to practical applications to NISQ quantum technologies, as they provide more efficient schemes for measuring observables at the output of noisy quantum devices. An interesting direction of future research is to explore scenarios where only a given set of physically relevant observables is concerned~\cite{aaronson2018shadow, huang2020predicting}. Our results also open up a systematic way to study new quantum protocols that sample over different transformations of quantum processes, known as quantum supermaps~\cite{chiribella2008transforming, chiribella2009theoretical, chiribella2013quantum}.

\section{Acknowledgments}
We thank Yin Mo and Chengkai Zhu for useful comments that helped us improve the manuscript.
G. C. has been supported by the Hong Kong Research Grant Council through Grants No. 17307520, No. R7035-21F, and No. T45-406/23-R, by the Ministry of Science and Technology through Grant No. 2023ZD0300600,
and by the John Templeton Foundation through Grant 62312, The Quantum Information Structure of Spacetime.
The opinions expressed in this publication are those of the authors and do not necessarily reflect the views of the John Templeton Foundation. Research at the Perimeter Institute is supported by the Government of Canada through the Department of Innovation, Science and Economic Development Canada and by the Province of Ontario through the Ministry of Research, Innovation and Science.
X. W. was supported by the National Key R\&D Program of China (Grant No. 2024YFE0102500), the Guangdong Provincial Quantum Science Strategic Initiative (Grant No.  GDZX2303007), the Guangdong Provincial Key Lab of Integrated Communication, Sensing and Computation for Ubiquitous Internet of Things (Grant No. 2023B1212010007), the Start-up Fund (Grant No. G0101000151) from The Hong Kong University of Science and Technology (Guangzhou), the Quantum Science Center of Guangdong-Hong Kong-Macao Greater Bay Area, and the Education Bureau of Guangzhou Municipality.

%%%%%%%%%%%%%%%%%%%%%%%%%%%%%%%%%%%%%%%%%
% Bibliography
%%%%%%%%%%%%%%%%%%%%%%%%%%%%%%%%%%%%%%%%%
\bibliography{references}

% %%%%%%%%%%%%%%%%%%%%%%%%%%%%%%%%%%%%%%%%%%%%%%
\appendix
\onecolumngrid

%%%%%%%%%%%%%%%%%%%%%%%%%%%%%%%%%%%%%%%%%
%%%%%%%%%%%%%%%%%%%%%%%%%%%%%%%%%%%%%%%%%
\section{Appendix A: Shadow Simulation Codes}\label{appsec:shadow_sim}
\renewcommand{\theequation}{A\arabic{equation}}
\setcounter{equation}{0}
\renewcommand{\theHequation}{A\arabic{equation}}
A general quantum supermap sending a quantum channel $\cN_{A\to B}$ to another quantum channel $\cM_{A'\to B'}$ can be realized by a quantum circuit where the input channel is inserted between an encoding channel $\cE_{A'\to AE}$ and a decoding channel $\cD_{EB\to B'}$ with a possible ancillary system $E$ connecting them~\cite{chiribella2008transforming}.
Every quantum supermap is associated with a bipartite quantum channel $\widehat S_{A'B \to AB'} \coloneqq (\cI_{A} \otimes \cD_{EB\to B'}  ) \circ (\cE_{A'\to AE} \otimes \cI_B)$, which is no-signaling from Bob to Alice, meaning that $\tr_{B'}[\widehat S  (\rho_{A'} \otimes \rho_B)]$ is independent of $\rho_B$, for every $\rho_{A'}$.

In the setting of shadow simulation, simulation codes are not restricted to quantum supermaps.
A shadow simulation code allows classical postprocessing so that its encoding and decoding parts do not have to be quantum channels. For example, one can consider the encoding operation $\widetilde{\cE}_{A'\to A} \coloneqq \sum_j \lambda_j \cE_{j, A'\to A}$ and the decoding operation $\widetilde{\cD}_{B\to B'} \coloneqq \sum_k \mu_k \cD_{k, B\to B'}$ with real coefficients $(\lambda_j)_j$ and $(\mu_k)_k$ and quantum channels $\lrp{\cE_{j, A'\to A}}_j$ and $\lrp{\cD_{k, B\to B'}}_k$.
The expectation value with respect to the state transmitted through a quantum channel $\cN_{A\to B}$ using this shadow simulation code $\lrp{\widetilde{\cE}_{A'\to A}, \widetilde{\cD}_{B\to B'}}$ can be decomposed as
\begin{align}
   \tr\left[\widetilde{\cD}_{B\to B'} \circ \cN_{A\to B} \circ \widetilde{\cE}_{A'\to A}(\rho_{RA'}) O_{RB'}\right] = \sum_{j,k} \lambda_j \mu_k \tr\left[\cD_{k, B\to B'} \circ \cN_{A\to B} \circ \cE_{j, A'\to A}(\rho_{RA'}) O_{RB'}\right].
\end{align}
Hence, although we cannot directly implement $\widetilde{\cE}_{A'\to A}$ and $\widetilde{\cD}_{B\to B'}$, we can simulate their effect by sending copies of the state using quantum simulation codes sampled from $\lrc{\lrp{\cE_{j, A'\to A}, \cD_{k, B\to B'}}}$ for multiple rounds and then postprocessing the measurement results from all rounds.

\begin{remark}
    Like how each quantum supermap is associated with a bipartite quantum channel, each virtual supermap is associated with a virtual channel. This is because a virtual supermap is essentially a linear combination of quantum supermaps, so its corresponding bipartite map is a linear combination of bipartite quantum channels, which is effectively a virtual channel.
\end{remark}

Note that the encoding operation $\widetilde{\cE}$ and the decoding operation $\widetilde{\cD}$ are both implemented with classical postprocessing, which only happens at Bob's side after Bob completes the measurement. Therefore, Bob needs information on Alice's sampled operation in each round to guide the postprocessing.
Hence, we consider simulation codes where Alice and Bob have pre-shared randomness, and such codes are realized by the following steps:
\begin{enumerate}
    \item Alice randomly samples one encoding channel $\cE_{j, A'\to A}$ from the set of channels $\left\{\cE_{j, A'\to A}\right\}$ with a probability distribution $\Pr\left(\cE_{j, A'\to A}\right) = |\lambda_j|/\gamma$, where $\gamma = \sum_{j} |\lambda_j|$, and applies the sampled channel to state $\rho_{RA'}$.
    \item Alice then sends the post-encoding state $(\cI_R \otimes \cE_{j, A'\to A})(\rho_{RA'})$ into the noisy channel $\cN_{A\to B}$.
    \item Upon receiving the state $(\cI_R \otimes \cN_{A\to B})\circ  (\cI_R \otimes \cE_{j, A'\to A})(\rho_{RA'})$ coming out of the noisy channel, Bob applies the decoding channel $\cD_{j, B\to B'}$ to the received state, where the value $j$ is known to Bob due to the classical randomness shared between him and Alice.
    \item Bob measures the decoded state $(\cI_R \otimes \cD_{j, B\to B'} \circ \cN_{A\to B} \circ \cE_{j, A'\to A})(\rho_{RA'})$ with an observable $O_{RB'}$, which gives a measurement outcome $o$.
    \item Repeat the above steps for $M$ times, and denote the index $j$ and the measurement outcome $o$ in the $m$-th round by $j_m$ and $o_m$, respectively. Then, compute the quantity $\xi \coloneqq \frac{\gamma}{M} \sum_{m=1}^M {\rm sign}(\lambda_{j_m})o_m$, where ${\rm sign}$ is the sign function.
\end{enumerate}
The quantity $\xi$ defined in the above protocol is an unbiased estimator for the expectation value $\tr\lrb{\lrp{\widetilde{\cS}\lrp{\cN}} (\rho) O}$, where $\widetilde{\cS}$ is the virtual supermap defined by
\begin{align}\label{appeq:virtual_supermap_action}
    \widetilde\cS(\cN)_{A'\to B'} = \sum_{j} \lambda_j \cD_{j, B\to B'} \circ \cN_{A\to B} \circ \cE_{j, A'\to A}.
\end{align}
Indeed, the expectation value of $\xi$ is
\begin{align}
    \mathbb{E}\lrb{\xi} &= \frac{\gamma}{M} \sum_{m=1}^M \mathbb{E}\lrb{{\rm sign}(\lambda_{j_m})o_m}\\
    &= \frac{\gamma}{M} \sum_{m=1}^M \sum_j \frac{\lrv{\lambda_j}}{\gamma} {\rm sign}(\lambda_j) \tr\lrb{\cD_{j} \circ \cN \circ \cE_j(\rho) O}\\
    &= \frac{1}{M} \sum_{m=1}^M \sum_{j} \lambda_j \tr\lrb{\cD_j \circ \cN \circ \cE_j(\rho) O}\\
    &= \tr\lrb{\sum_j \lambda_j \cD_j \circ \cN \circ \cE_j(\rho) O}\\
    &= \tr\lrb{\lrp{\widetilde\cS(\cN)} (\rho) O}.
\end{align}

\begin{remark}
    As related works, both Ref.~\cite{yuan2023virtual} and Ref.~\cite{liu2024virtual} discuss virtual distillation of quantum channels, which is a special case of our shadow simulation task.
    Ref.~\cite{yuan2023virtual} focuses mostly on static resources and, for virtual distillation of channels, only considers the scenario where Bob alone performs local virtual operations.
    Ref.~\cite{liu2024virtual} proposed a fixed circuit acting simultaneously on multiple copies of an unknown noisy quantum gate as a channel distillation protocol in the low-noise regime, while our approach involves converting a known channel into another known channel through sampling over different quantum circuits, and is not limited to the low-noise regime.
\end{remark}

Each supermap is associated with a bipartite map, and for a randomness-assisted shadow simulation code in Eq.~\eqref{appeq:virtual_supermap_action}, its corresponding bipartite map is
\begin{align}
    \widehat{\cS}_{A'B\to AB'} = \sum_j \lambda_j \cE_{j, A'\to A} \ox \cD_{j, B\to B'}.
\end{align}
It is clear that every randomness-assisted shadow simulation code is associated with a Hermitian-preserving bipartite map. This is because the Choi operator of the associated bipartite map is Hermitian, and a map is Hermitian-preserving if and only if its Choi operator is Hermitian.

In the following, we show that bipartite maps associated with randomness-assisted shadow simulation codes are no-signaling.
For a bipartite map $\widehat\cS_{A'B\to AB'}$, being no-signaling from Bob to Alice means that the output state at $A$ is independent of the input state at $B$, {\em i.e.}, $\tr_{B'} \circ \widehat\cS_{A'B\to AB'} = \widehat\cS_{A'\to A} \circ \tr_{B}$, where $\widehat\cS_{A'\to A}$ is the local effective operation of the map $\widehat\cS_{A'B\to AB'}$ from $A$ to $A'$.
By the \Choi isomorphism~\cite{jamiolkowski1972linear, choi1975completely}, we can uniquely represent $\widehat\cS_{A'B\to AB'}$ using its Choi operator $J^{\widehat\cS}_{A'BAB'} \coloneqq \cI_{A'B} \ox \widehat\cS_{\bar{A'}\bar{B}\to AB'} (\proj{\Gamma}_{A'B\bar{A'}\bar{B}})$,
where $\ket{\Gamma}_{A'B\bar{A'}\bar{B}} \coloneqq \sum_{j=0}^{d_{A'B}-1} \ket{j}_{A'B}\ket{j}_{\bar{A'}\bar{B}}$ is the unnormalized maximally entangled state, $d_{A'B}$ is the dimension of the system $A'B$, and the Hilbert space associated with the system $\bar{A'}\bar{B}$ is isomorphic to that of $A'B$.
In terms of the map's Choi operator, no-signaling from Bob to Alice means $\tr_{B'}\lrb{J^{\widehat\cS}_{A'BAB'}} = J^{\widehat\cS}_{A'A} \ox \idop_B$,
where $J^{\widehat\cS}_{A'A} \coloneqq \tr_{BB'}[J^{\widehat\cS}_{A'BAB'}] / d_B$.

Theorem~\ref{thm:ns_code} shows that randomness-assisted shadow simulation codes are equivalent to Hermitian-preserving no-signaling bipartite maps, which is both no-signaling from Bob to Alice and no-signaling from Alice to Bob.

\renewcommand\theproposition{\ref{thm:ns_code}}
\setcounter{proposition}{\arabic{proposition}-1}
\begin{theorem}
    A virtual supermap $\widetilde \cS$ is a randomness-assisted shadow simulation code if and only if the corresponding virtual channel $\widehat \cS$ is no-signaling.
\end{theorem}
\renewcommand{\theproposition}{\arabic{proposition}}
\begin{proof}
    For the ``only if'' part, let $\widehat\cS_{A'B\to AB'} = \sum_j \lambda_j \cM^{(j)}_{A'\to A} \ox \cN^{(j)}_{B\to B'}$ denote a bipartite map associated with an arbitrary randomness-assisted code, where $\cM^{(j)}_{A'\to A}$ and $\cN^{(j)}_{B\to B'}$ are quantum channels. It is clear that the Choi operator $J^{\widehat\cS}_{A'BAB'} = \sum_j \lambda_j J^{\cM^{(j)}}_{A'A} \ox J^{\cN^{(j)}}_{BB'}$ is Hermitian, indicating that $\widehat\cS_{A'B\to AB'}$ is Hermitian-preserving. Furthermore, as $\cM^{(j)}_{A'\to A}$ and $\cN^{(j)}_{B\to B'}$ are trace-preserving, {\em i.e.}, $\tr_A\lrb{J^{\cM^{(j)}}_{A'A}} = \idop_{A'}$ and $\tr_{B'}\lrb{J^{\cN^{(j)}}_{BB'}} = \idop_B$, we have
    \begin{align}
        \tr_A\lrb{J^{\widehat\cS}_{A'BAB'}} &= \sum_j \lambda_j \tr_A\lrb{J^{\cM^{(j)}}_{A'A}} \ox J^{\cN^{(j)}}_{BB'} = \idop_{A'} \ox \sum_j \lambda_j J^{\cN^{(j)}}_{BB'}
    \end{align}
    and
    \begin{align}
        \tr_{B'}\lrb{J^{\widehat\cS}_{A'BAB'}} &= \sum_j \lambda_j J^{\cM^{(j)}}_{A'A} \ox \tr_{B'}\lrb{J^{\cN^{(j)}}_{BB'}} = \sum_j \lambda_j J^{\cM^{(j)}}_{A'A} \ox \idop_B,
    \end{align}
    which imply that $\widehat\cS_{A'B\to AB'}$ is no-signaling.

    For the ``if'' part, we first assume that $\widehat\cS_{A'B\to AB'}$ is Hermitian-preserving and no-signaling. By Theorem 14 in Ref.~\cite{gutoski2009properties}, its Choi operator can be decomposed as
    \begin{align}
        J^{\widehat\cS}_{A'BAB'} = \sum_{j} \lambda_j M^{(j)}_{A'A} \ox N^{(j)}_{BB'},
    \end{align}
    where, for each $j$, $\lambda_j$ is a real number, and $M^{(j)}_{A'A}$ and $N^{(j)}_{BB'}$ are Hermitian operators such that $\tr_A\lrb{M^{(j)}_{A'A}}$ and $\tr_{B'}\lrb{N^{(j)}_{BB'}}$ are proportional to the identity operators $\idop_{A'}$ and $\idop_B$, respectively. In other words, we can treat $M^{(j)}_{A'A}$ and $N^{(j)}_{BB'}$ as Choi operators for some Hermitian-preserving and trace-scaling (HPTS) maps $\widetilde{\cM}^{(j)}_{A'\to A}$ and $\widetilde{\cN}^{(j)}_{B\to B'}$, respectively. Here, trace-scaling means that the map scales the trace of the input operator with a constant factor.

    According to Lemma 6 in Ref.~\cite{zhao2022information}, every HPTS map can be written as a linear combination of two quantum channels. Thus, we can write $\widehat\cS_{A'B\to AB'}$ as
    \begin{align}
        \widehat\cS_{A'B\to AB'} &= \sum_{j} \lambda_j \widetilde{\cM}^{(j)}_{A'\to A} \ox \widetilde{\cN}^{(j)}_{B\to B'} \\
        &= \sum_{j} \lambda_j \lrp{m_1\cM^{(j,1)}_{A'\to A} + m_2\cM^{(j,2)}_{A'\to A}} \ox \lrp{n_1\cN^{(j,1)}_{B\to B'} + n_2\cN^{(j,2)}_{B\to B'}}\\
        &= \sum_j\sum_{k=1}^2\sum_{l=1}^2 \lambda_j m_k n_l \cM^{(j,k)}_{A'\to A} \ox \cN^{(j,l)}_{B\to B'},
    \end{align}
    which represents a randomness-assisted shadow simulation code.
\end{proof}

Alternative to pre-shared classical randomness, forward classical communication from Alice to Bob also allows Bob to acquire information on Alice's local operation in each round.
A shadow simulation protocol with the assistance of forward classical communication is represented by a bipartite linear map
\begin{align}\label{eq:fc_code}
    \widehat{\cS}_{A'B\to AB'} = \sum_j \lambda_j \cM^{(j)}_{A'\to A} \ox \cN^{(j)}_{B\to B'},
\end{align}
where $\lrc{\cM^{(j)}_{A'\to A}}_j$ is a quantum instrument, $\lrc{\cN^{(j)}_{B\to B'}}_j$ is a collection of quantum channels, and each $\lambda_j$ is a real coefficient.
In the following proposition, we show that not only $\widehat{\cS}_{A'B\to AB'}$ is Hermitian-preserving and one-way no-signaling, but any one-way no-signaling Hermitian-preserving bipartite map represents a forward-classical-assisted shadow simulation code.
For completeness, we also consider shadow simulation codes assisted by two-way classical communication, where both $\lrc{\cM^{(j)}_{A'\to A}}_j$ and $\lrc{\cN^{(j)}_{B\to B'}}_j$ are quantum instruments. Such codes are equivalent to the set of all bipartite Hermitian-preserving maps.

\begin{theorem}\label{prop:classical_comm_code}
    Consider a bipartite linear map $\widehat\cS_{A'B\to AB'}$. It is Hermitian-preserving if and only if it corresponds to a shadow simulation code assisted by two-way classical communication. It is Hermitian-preserving and $B$-to-$A$ no-signaling if and only if it corresponds to a forward-classical-assisted shadow simulation code.
\end{theorem}

This theorem tells us that shadow simulation codes with one-way classical communication are powerful enough to simulate arbitrary quantum channels and even beyond. An intuitive explanation is that Alice can measure the initial state $\rho$ with an informationally complete POVM and communicate the measurement outcomes to Bob so that Bob is able to reconstruct the expectation value of every observable on $\rho$ transformed by any channel. Now, we prove this theorem by proving the following two lemmas.

\begin{lemma}\label{lemma:fc_code}
    A bipartite linear map $\widehat\cS_{A'B\to AB'}$ is Hermitian-preserving and $B$-to-$A$ no-signaling if and only if it corresponds to a forward-classical-assisted shadow simulation code.
\end{lemma}
\begin{proof}
    The ``if'' direction is straightforward. Let $\widehat\cS_{A'B\to AB'} = \sum_j \lambda_j \cM^{(j)}_{A'\to A} \ox \cN^{(j)}_{B\to B'}$ represent a forward-classical-assisted shadow simulation code, where each $\lrc{\cM^{(j)}_{A'\to A}}_j$ is a quantum instrument, and each $\cN^{(j)}_{B\to B'}$ is a quantum channel. Then, the Choi operator of $\widehat\cS_{A'B\to AB'}$ is
    \begin{align}
        J^{\widehat\cS}_{A'BAB'} = \sum_j \lambda_j J^{\cM^{(j)}}_{A'A} \ox J^{\cN^{(j)}}_{BB'}.
    \end{align}
    Clearly, $J^{\widehat\cS}_{A'BAB'}$ is a Hermitian operator, indicating that $\widehat\cS_{A'B\to AB'}$ is Hermitian-preserving. In addition, $\widehat\cS_{A'B\to AB'}$ is no-signaling from Bob to Alice as
    \begin{align}
        \tr_{B'}\lrb{J^{\widehat\cS}_{A'BAB'}} = \sum_j \lambda_j J^{\cM^{(j)}}_{A'A} \ox \tr_{B'}\lrb{J^{\cN^{(j)}}_{BB'}} = \sum_j \lambda_j J^{\cM^{(j)}}_{A'A} \ox \idop_{B}
    \end{align}
    due to $\cN$ being trace-preserving.

    For the ``only if'' direction, part of the proof is adapted from the proof of Theorem 14 in Ref.~\cite{gutoski2009properties}. Let $\widehat\cS_{A'B\to AB'}$ be a Hermitian-preserving supermap and $\lrc{M^{(1)}_{A'A}, \dots, M^{(D)}_{A'A}}$ be a basis for the Hermitian operator space ${\rm Herm}\lrp{A'A}$ on the system $A'A$, where $D$ is the dimension of this space. Then, there exists a unique set of Hermitian operators $\lrc{N^{(1)}_{BB'}, \dots, N^{(D)}_{BB'}} \subseteq {\rm Herm}\lrp{BB'}$ such that $J^{\widehat\cS}_{A'BAB'} = \sum_{j=1}^D M^{(j)}_{A'A} \ox N^{(j)}_{BB'}$.
    For each $j$, let $H^{(j)}_{A'A}$ be a Hermitian operator such that $\tr\lrb{H^{(j)}_{A'A} M^{(k)}_{A'A}} \neq 0$ if and only if $j = k$. Denoting the mapping $M_{A'A} \mapsto \tr\lrb{H^{(j)}_{A'A} M^{(k)}_{A'A}}$ by $h^{(j)}(M_{A'A})$, we have
    \begin{align}
        \lrp{h^{(j)} \ox \cI_{BB'}} \lrp{J^{\widehat\cS}_{A'BAB'}} = \sum_{k=1}^D h^{(j)}\lrp{M^{(k)}_{A'A}} \ox N^{(k)}_{BB'} = h^{(j)}\lrp{M^{(j)}_{A'A}} N^{(j)}_{BB'}.
    \end{align}
    Because $\widehat\cS_{A'B\to AB'}$ is no-signaling from Bob to Alice, we have $\tr_{B'}\lrb{J^{\widehat\cS}_{A'BAB'}} = J^{\widehat\cS}_{A'A} \ox \idop_B$. Then,
    \begin{align}\label{eq:fc-assisted-proof_prop-to-identity}
        \lrp{h^{(j)} \ox \cI_{BB'}} \tr_{B'}\lrb{J^{\widehat\cS}_{A'BAB'}} = h^{(j)} \lrp{J^{\widehat\cS}_{A'A}} \ox \idop_B = h^{(j)}\lrp{M^{(j)}_{A'A}} \tr_{B'}\lrb{N^{(j)}_{BB'}},
    \end{align}
    where the second equality holds because the order of applying $\lrp{h^{(j)} \ox \cI_{BB'}}$ and $\tr_{B'}$ does not affect the result as they act on different subspaces.
    
    It follows from Eq.~\eqref{eq:fc-assisted-proof_prop-to-identity} that, for each $j$, $\tr_{B'}\lrb{N^{(j)}_{BB'}}$ is proportional to the identity operator $\idop_B$, and thus it serves as a Choi operator of an HPTS map, which we denote by $\widetilde{\cN}^{(j)}_{B\to B'}$. According to Lemma 6 in Ref.~\cite{zhao2022information}, each $\widetilde{\cN}^{(j)}_{B\to B'}$ can be written as a linear combination of two quantum channels, {\em i.e.},
    \begin{align}
        \widetilde{\cN}^{(j)}_{B\to B'} = n_{j,1} \cN^{(j,1)}_{B\to B'} + n_{j,2} \cN^{(j,2)}_{B\to B'},
    \end{align}
    where $n_{j,1}$ and $n_{j,2}$ are real numbers and $\cN^{(j,1)}_{B\to B'}$ and $\cN^{(j,2)}_{B\to B'}$ are quantum channels.

    For each Hermitian operator $M^{(j)}_{A'A}$, we can write it as the difference of two positive semidefinite operators, say, $m_{j,1}M^{(j,1)}_{A'A} - m_{j,2}M^{(j,2)}_{A'A}$, where $m_{j,1}M^{(j,1)}_{A'A}$ and $m_{j,2}M^{(j,2)}_{A'A}$ are positive semidefinite and $m_{j,1}$ and $m_{j,2}$ are positive real numbers so that $\tr_{A}\lrb{M^{(j,1)}_{A'A}} \leq \idop_{A'}$ and $\tr_{A}\lrb{M^{(j,2)}_{A'A}} \leq \idop_{A'}$. In other words, both $M^{(j,1)}_{A'A}$ and $M^{(j,2)}_{A'A}$ are Choi operators of some completely positive and trace-non-increasing (CPTN) maps, say $\cM^{(j,1)}_{A'\to A}$ and $\cM^{(j,2)}_{A'\to A}$, respectively. Moreover, the scalars $m_{j,1}$ and $m_{j,2}$ should be chosen so that $\sum_{j=1}^D \sum_{k=1}^2 \tr_{A}\lrb{M^{(j,k)}_{A'A}} \leq \idop_{A'} / 2$. We will see the reason of this requirement later.

    Combining the decomposition of every $M^{(j)}_{A'A}$ and every $N^{(j)}_{BB'}$, we have
    \begin{align}
        \cS_{A'B\to AB'} &= \sum_{j=1}^D \lrp{m_{j,1}\cM^{(j,1)}_{A'\to A} - m_{j,2}\cM^{(j,2)}_{A'\to A}} \ox \lrp{n_{j,1}\cN^{(j,1)}_{B\to B'} + n_{j,2}\cN^{(j,2)}_{B\to B'}}\\
        &= \sum_{j=1}^D \Big(m_{j,1}n_{j,1}\cM^{(j,1)}_{A'\to A} \ox \cN^{(j,1)}_{B\to B'} + m_{j,1}n_{j,2}\cM^{(j,1)}_{A'\to A} \ox \cN^{(j,2)}_{B\to B'} \nonumber\\
        &\qquad - m_{j,2}n_{j,1}\cM^{(j,2)}_{A'\to A} \ox \cN^{(j,1)}_{B\to B'} - m_{j,2}n_{j,2}\cM^{(j,2)}_{A'\to A} \ox \cN^{(j,2)}_{B\to B'}\Big)\\
        &= \sum_{j=1}^{4D} \lambda_j \cM^{(j)}_{A'\to A} \ox \cN^{(j)}_{B\to B'}
    \end{align}
    with appropriate relabeling, where each $\lambda_j \in \lrc{m_{j,1}n_{j,1}, m_{j,1}n_{j,2}, - m_{j,2}n_{j,1}, - m_{j,2}n_{j,2}}_j$ is a real number, each $\cM^{(j)}_{A'\to A} \in \lrc{\cM^{(j,1)}_{A'\to A}, \cM^{(j,2)}_{A'\to A}}_j$ is a CPTN map, and each $\cN^{(j)}_{B\to B'} \in \lrc{\cN^{(j,1)}_{B\to B'}, \cN^{(j,2)}_{B\to B'}}_j$ is a quantum channel. Note that
    \begin{align}
        \sum_{j=1}^{4D} \tr_{A}\lrb{J^{\cM^{(j)}}_{A'A}} = 2 \sum_{j=1}^D \sum_{k=1}^2 \tr_{A}\lrb{J^{\cM^{(j,k)}}_{A'A}} \leq \idop_{A'}
    \end{align}
    due to our choice of scalars $\lrc{m_{j,1}, m_{j,2}}_j$. Let $\cM'_{A'\to A}$ be a CPTN map such that $\sum_{j=1}^{4D} \tr_{A}\lrb{J^{\cM^{(j)}}_{A'A}} + 2\tr_{A}\lrb{J^{\cM'}_{A'A}} = \idop_{A'}$. Then, $\lrc{\cM^{(1)}, \dots, \cM^{(4D)}, \cM', \cM'}$ is a quantum instrument and
    \begin{align}
        \widehat\cS_{A'B\to AB'} = \sum_{j=1}^{4D} \lambda_j \cM^{(j)}_{A'\to A} \ox \cN^{(j)}_{B\to B'} + \cM'_{A'\to A} \ox \cN'_{B\to B'} - \cM'_{A'\to A} \ox \cN'_{B\to B'}
    \end{align}
    for any quantum channel $\cN'_{B\to B'}$. Therefore, any bipartite linear map $\cS_{A'B\to AB'}$ that is Hermitian-preserving and $B$-to-$A$ no-signaling represents a forward-classical-assisted shadow simulation code.
\end{proof}

\begin{lemma}\label{lemma:2wc_code}
    A bipartite linear map $\widehat\cS_{A'B\to AB'}$ is Hermitian-preserving if and only if it corresponds to a shadow simulation code assisted by two-way classical communication.
\end{lemma}
\begin{proof}
    The ``if'' part can be directly verified by checking that the Choi operator of a bipartite map $\sum_j \lambda_j \cM^{(j)}_{A'\to A} \ox \cN^{(j)}_{B\to B'}$ is Hermitian, where both $\lrc{\cM^{(j)}_{A'\to A}}_j$ and $\lrc{\cN^{(j)}_{B\to B'}}_j$ are quantum instruments.

    For the ``only if'' part, we follow the proof of Lemma~\ref{lemma:fc_code} to write the Choi operator of a bipartite Hermitian-preserving map $\widehat\cS_{A'B\to AB'}$ as $J^{\widehat\cS}_{A'BAB'} = \sum_{j=1}^D M^{(j)}_{A'A} \ox N^{(j)}_{BB'}$, where $M^{(j)}_{A'A}$ and $N^{(j)}_{BB'}$ are Hermitian operators.
    Each $M^{(j)}_{A'A}$ or $N^{(j)}_{BB'}$ can be written as the difference of two positive semidefinite operators. We write each $M^{(j)}_{A'A}$ as $m_{j,+}M^{(j,+)}_{A'A} - m_{j,-}M^{(j,-)}_{A'A}$ and each $N^{(j)}_{BB'}$ as $n_{j,+}N^{(j,+)}_{BB'} - n_{j,-}N^{(j,-)}_{BB'}$, where $m_{j,\pm}M^{(j,1)}_{A'A}, n_{j,2}N^{(j,\pm)}_{BB'}$ are positive semidefinite and $m_{j,\pm}, n_{j,\pm}$ are positive real numbers that will be fixed later. The Choi operator $J^{\widehat\cS}_{A'BAB'}$ now can be written as
    \begin{align}
        J^{\widehat\cS}_{A'BAB'} &= \sum_{j=1}^D \lrp{m_{j,+}M^{(j,+)}_{A'A} - m_{j,-}M^{(j,-)}_{A'A}} \ox \lrp{n_{j,+}N^{(j,+)}_{BB'} - n_{j,-}N^{(j,-)}_{BB'}}\\
        &= \sum_{j=1}^D \sum_{k \in \lrc{+,-}} \sum_{l \in \lrc{+,-}} (-1)^{1 - \delta_{k,l}} m_{j,k}n_{j,l} M^{(j,k)}_{A'A} \ox N^{(j,l)}_{BB'},
    \end{align}
    where $\delta_{k,l} = 1$ if $k = l$ and $\delta_{k,l} = 0$ otherwise.
    Because $M^{(j,\pm)}_{A'A}$ and $N^{(j,\pm)}_{BB'}$ are positive semidefinite operators, they can be treated as Choi operators of completely positive maps, say, $\cM^{(j,\pm)}_{A'\to A}$ and $\cN^{(j,\pm)}_{B\to B'}$.
    Then, we can write $\widehat\cS_{A'B\to AB'}$ as
    \begin{align}
        \widehat\cS_{A'B\to AB'} &= \sum_{j=1}^D \sum_{k \in \lrc{+,-}} \sum_{l \in \lrc{+,-}} (-1)^{1 - \delta_{k,l}} m_{j,k}n_{j,l} \cM^{(j,k)}_{A'\to A} \ox \cN^{(j,l)}_{B\to B'}\\
        &= \sum_{j=1}^{4D} \lambda_j \cM^{(j)}_{A'\to A} \ox \cN^{(j)}_{B\to B'}
    \end{align}
    with appropriate relabeling, where each $\lambda_j \in \lrc{(-1)^{1 - \delta_{k,l}} m_{j,k}n_{j,l}}_{j,k,l}$ is a real number, each $\cM^{(j)}_{A'\to A} \in \lrc{\cM^{(j,\pm)}_{A'\to A}}_j$ or $\cN^{(j)}_{B\to B'} \in \lrc{\cN^{(j,\pm)}_{B\to B'}}_j$ is a CPTN map.
    We can fix the values of the coefficients $m_{j,\pm}$ and $n_{j,\pm}$ to be large enough so that both $\sum_{j=1}^{4D} \cM^{(j)}_{A'\to A}$ and $\sum_{j=1}^{4D} \cN^{(j)}_{B\to B'}$ are trace-non-increasing. Let $\cM'_{A'\to A}$ and $\cN'_{B\to B'}$ be CPTN maps such that $\sum_{j=1}^{4D} \cM^{(j)}_{A'\to A} + 2\cM'_{A'\to A}$ and $\sum_{j=1}^{4D} \cN^{(j)}_{B\to B'} + 2\cN'_{B\to B'}$ are CPTP. That is, $\lrc{\cM^{(1)}, \dots, \cM^{(4D)}, \cM', \cM'}$ and $\lrc{\cN^{(1)}, \dots, \cN^{(4D)}, \cN', \cN'}$ are quantum instruments. Because we can write
    \begin{align}
        \widehat\cS_{A'B\to AB'} = \sum_{j=1}^{4D} \lambda_j \cM^{(j)}_{A'\to A} \ox \cN^{(j)}_{B\to B'} + \cM'_{A'\to A} \ox \cN'_{B\to B'} - \cM'_{A'\to A} \ox \cN'_{B\to B'},
    \end{align}
    it follows that $\widehat\cS_{A'B\to AB'}$ corresponds to a shadow simulation code assisted by two-way classical communication. Hence the proof.
\end{proof}

From now on, we focus on no-signaling shadow simulation codes. We consider implementing such codes by sampling quantum no-signaling codes. This is possible due to the following proposition.

\begin{proposition}
    A bipartite linear map is Hermitian-preserving and no-signaling if and only if it is a linear combination of bipartite linear maps that correspond to quantum no-signaling codes.
\end{proposition}
\begin{proof}
    For the ``if'' direction, let $\widehat{\cS}_{A'B\to AB'} = \sum_j \cS^{(j)}_{A'B\to AB'}$ be a linear combination of bipartite linear maps that correspond to quantum no-signaling codes. The map $\widehat{\cS}_{A'B\to AB'}$ is Hermitian-preserving because $J^{\widehat{\cS}}_{A'BAB'}$, the Choi operator of $\widehat{\cS}_{A'B\to AB'}$, is Hermitian. Also, $\widehat{\cS}_{A'B\to AB'}$ is no-signaling from $B$ to $A$, because
    \begin{align}
        \tr_{B'}\lrb{J^{\widehat{\cS}}_{A'BAB'}} = \sum_j \lambda_j \tr_{B'}\lrb{J^{\cS^{(j)}}_{A'BAB'}} = \sum_j \lambda_j J^{\cS^{(j)}}_{A'A} \ox \idop_B = J^{\widehat{\cS}}_{A'A} \ox \idop_B,
    \end{align}
    where the second inequality follows from each $\cS^{(j)}_{A'B\to AB'}$ being no-signaling.
    Similarly, $\widehat{\cS}_{A'B\to AB'}$ is no-signaling from $A$ to $B$ as
    \begin{align}
        \tr_A\lrb{J^{\widehat{\cS}}_{A'BAB'}} = \sum_j \lambda_j \tr_A\lrb{J^{\cS^{(j)}}_{A'BAB'}} = \sum_j \lambda_j J^{\cS^{(j)}}_{BB'} \ox \idop_{A'} = J^{\widehat{\cS}}_{BB'} \ox \idop_{A'}.
    \end{align}
    Therefore, the map $\widehat{\cS}_{A'B\to AB'}$ is Hermitian-preserving and no-signaling.

    For the ``only if'' part, let $\widetilde{\cS}_{A'B\to AB'}$ be a Hermitian-preserving and no-signaling bipartite linear map. According to Theorem~\ref{thm:ns_code}, $\widehat{\cS}_{A'B\to AB'}$ can be decomposed as $\widehat{\cS}_{A'B\to AB'} = \sum_j \lambda_j \cM^{(j)}_{A'\to A} \ox \cN^{(j)}_{B\to B'}$ for some quantum channels $\cM^{(j)}_{A'\to A}$ and $\cN^{(j)}_{B\to B'}$. Note that each $\cS^{(j)}_{A'B\to AB'} \coloneqq \cM^{(j)}_{A'\to A} \ox \cN^{(j)}_{B\to B'}$ is a bipartite linear map corresponding to a quantum no-signaling code. Therefore, $\widehat{\cS}_{A'B\to AB'} = \sum_j \lambda_j \cS^{(j)}_{A'B\to AB'}$ is indeed a linear combination of bipartite linear maps corresponding to quantum no-signaling codes.
\end{proof}

By decomposing it into a few quantum no-signaling codes, we can implement any no-signaling shadow simulation code by sampling quantum no-signaling codes in a way similar to the protocol given earlier in this section for realizing randomness-assisted shadow simulation codes.
The implementation of a no-signaling shadow simulation code incurs a cost quantifying how many sampling rounds are required. Such a cost can be derived from Hoeffding's inequality. Let $\widetilde{\cS} = \sum_j \lambda_j \cS_j$ be a no-signaling shadow simulation code decomposed into a linear combination of quantum no-signaling codes $\lrc{\cS_j}$ so that
\begin{align}
    \tr\lrb{\lrp{\widetilde{\cS}\lrp{\cN}} (\rho)O} = \sum_j \lambda_j \tr\lrb{\lrp{\cS_j\lrp{\cN}} (\rho)O}
\end{align}
for any quantum state $\rho$ and any observable $O$.
We assume that the observable is bounded as $\Vert O \Vert_\infty \leq 1$ so that each measurement outcome belongs to the interval $[-1, 1]$. For postprocessing, we multiply each measurement outcome by a factor of magnitude $\gamma \coloneqq \sum_j |\lambda_j|$, and the average of all the postprocessed outcomes serves as an unbiased estimator $\xi$ for $\tr\lrb{\lrp{\widetilde{\cS}\lrp{\cN}} (\rho)O}$. According to Hoeffding's inequality~\cite{hoeffding1994probability}, the probability that the estimator has an error larger than or equal to $\epsilon$ is bounded as
\begin{align}
    \Pr\lrp{\lrv{\xi - \mathbb{E}\lrb{\xi}} \geq \epsilon} \leq 2 \exp \lrp{-\frac{M \epsilon^2}{2 \gamma^2}},
\end{align}
where $M$ is the number of sampling rounds.
Hence, we can conclude that
\begin{align}
    M \geq \frac{2\gamma^2 \log\frac{2}{\delta}}{\epsilon^2}
\end{align}
rounds are enough for the final estimation to have an error smaller than $\epsilon$ with a probability no less than $1 - \delta$.

The number of rounds $M$ is proportional to $\gamma^2$, where $\gamma$ is the sum of the absolute values of the coefficients in the decomposition of $\widetilde{\cS}$. Considering that a no-signaling shadow simulation code $\widetilde{\cS}$ can have many different decompositions, we define its sampling cost as the smallest possible $\gamma$ achieved by any feasible decomposition:
\begin{align}
    c_{\rm smp}\lrp{\widetilde{\cS}} \coloneqq \inf \lrc{ \sum_j \lrv{\lambda_j} ~\middle|~ \widehat \cS = \sum_j \lambda_j \widehat \cS_j,~ \lambda_j \in \mathbb{R},~ \widehat \cS_j \in \rm{CPTP} \cap \rm{NS} }.
\end{align}
Note that all the quantum no-signaling channels in the decomposition whose corresponding coefficients have the same sign can be grouped into one single quantum no-signaling channel without changing the cost. Hence, it is sufficient to consider all combinations in the form of $\widehat \cS = p_+ \widehat\cS^+ - p_- \widehat\cS^-$, where $p_\pm$ are non-negative coefficients and $\widehat\cS^\pm$ are quantum no-signaling channels:
\begin{align}
    c_{\rm smp}\lrp{\widetilde{\cS}} = \inf \lrc{p_+ + p_- ~\middle|~ \widehat \cS = p_+ \widehat\cS^+ - p_- \widehat\cS^-,~ p_\pm \in \mathbb{R}^+, \widehat\cS_j \in \rm{CPTP} \cap \rm{NS}}.
\end{align}
Note that every conventional channel simulation protocol has a sampling cost of $1$, while a shadow simulation protocol can have a sampling cost either larger or smaller than $1$, in addition to being equal to $1$.

Besides sampling cost, simulation error is an important indicator on the performance of a simulation code. In the main text, we use diamond distance between the simulated map $\widetilde{\cM} \coloneqq \widetilde{\cS}(\cN)$ and the target channel $\cM$ to measure this error. Here, we justify our choice.

As the target of our task is to estimate the expectation value, the most direct measure of the error is
\begin{align}
    \lrv{\tr\lrb{O_{RB'}\widetilde{\cM}_{A'\to B'}(\rho_{RA'})} - \tr\lrb{O_{RB'}\cM_{A'\to B'}(\rho_{RA'})}}
\end{align}
for some given observable $O_{RB'}$ and quantum state $\rho_{RA'}$, where $R$ is some reference system inaccessible to Alice.
Because one simulation code should work for every quantum state and every observable, we consider the worst case error, which maximizes the error over all quantum states and observables.
Without loss of generality, we consider only observables with $\lrV{O}_\infty \leq 1$ since every other observable is such an observable multiplied by a scalar:
\begin{align}
    \sup_{\rho_{RA'},~ O_{RB'}:\lrV{O}_\infty \leq 1} \lrv{\tr[O\widetilde{\cM}(\rho)] - \tr[O\cM(\rho)]}.
\end{align}
Observe that the worst case error is upper-bounded by diamond norm because for any observable $O$ and state $\rho$, we have
\begin{align}
    \lrv{\tr[O\widetilde{\cM}(\rho)] - \tr[O\cM(\rho)]} &\leq \lrv{\tr[O^+(\widetilde{\cM}(\rho) - \cM(\rho))]} + \lrv{\tr[O^-(\widetilde{\cM}(\rho) - \cM(\rho))]}\\
    &\leq \frac{1}{2} \lrV{\widetilde{\cM}(\rho) - \cM(\rho)}_1 + \frac{1}{2} \lrV{\widetilde{\cM}(\rho) - \cM(\rho)}_1\\
    &= \lrV{\widetilde{\cM}(\rho) - \cM(\rho)}_1\\
    &\leq \lrV{\widetilde{\cM} - \cM}_\diamond,
\end{align}
where $O^+$ and $O^-$ are positive semidefinite operators representing the positive and negative parts of $O$, respectively.
Furthermore, this upper bound is tight in the sense that there always exists an observable $O$ and a quantum state $\rho$ that saturate this bound.
Specifically, let $\rho^*$ be a quantum state such that $\lrV{\widetilde{\cM}(\rho^*) - \cM(\rho^*)}_1 = \lrV{\widetilde{\cM} - \cM}_\diamond$ and $O^* \geq 0$ be an observable such that $\lrv{\tr\lrb{O^*\lrp{\widetilde{\cM}(\rho^*) - \cM(\rho^*)}}} = \frac{1}{2} \lrV{\widetilde{\cM}(\rho^*) - \cM(\rho^*)}$. Such state $\rho^*$ and observable $O^*$ always exist, and they saturate the upper bound, that is
\begin{align}
    \sup_{\rho,~ O:\lrV{O}_\infty \leq 1} \lrv{\tr\lrb{O\widetilde{\cM}(\rho)} - \tr\lrb{O\cM(\rho)}} &= \lrv{\tr\lrb{O^*\widetilde{\cM}(\rho^*)} - \tr\lrb{O^*\cM(\rho^*)}}\\
    &= \lrV{\widetilde{\cM} - \cM}_\diamond.
\end{align}
Therefore, diamond distance is indeed the worst case error of estimating expectation values.

%%%%%%%%%%%%%%%%%%%%%%%%%%%%%%%%%%%%%%%%%
%%%%%%%%%%%%%%%%%%%%%%%%%%%%%%%%%%%%%%%%%
\section{Appendix B: General SDPs for Minimum Error and Minimum Sampling Cost}
\renewcommand{\theequation}{B\arabic{equation}}
\setcounter{equation}{0}
\renewcommand{\theHequation}{B\arabic{equation}}
We show that the minimum sampling cost and the minimum error of shadow simulation assisted by no-signaling codes can be formulated as SDPs. The minimum sampling cost can be formulated as
\begin{subequations}
\begin{align}
    \gamma^*_{\varepsilon, \rm NS}(\cN, \cM) = \inf &\; p_+ + p_-\\
    \text{s.t.} &\; \widehat\cS = p_+ \widehat\cS^+ - p_- \widehat\cS^-,\\
    &\; \frac{1}{2} \lrV{\cM_{A'\to B'} - \widetilde{\cS}\lrp{\cN_{A\to B}}}_\diamond \leq \varepsilon,\\
    &\; \widehat\cS^\pm \in \rm{CPTP} \cap {\rm NS}.
\end{align}
\end{subequations}
This optimization problem can be modified to one for $\varepsilon^*_{\gamma, \rm NS}\lrp{\cN, \cM}$ by changing the optimization objective to $\varepsilon$ and adding the constraint $p_+ + p_- \leq \gamma$.
For a pair of quantum channels, {\em i.e.}, CPTP maps, the diamond distance between them can be efficiently computed via a simple SDP~\cite{watrous2009semidefinite}. For shadow simulation, however, the map $\widetilde{\cS}\lrp{\cN}$ is HPTS, which is more general than CPTP. Here, we show how to adapt the SDP for the diamond distance between two quantum channels to compute the diamond distance between any two HPTS maps.

Let $\widetilde{\cN}_{A\to B} = p_+\cN^+_{A\to B} - p_-\cN^-_{A\to B}$ and $\widetilde{\cM}_{A\to B} = q_+\cM^+_{A\to B} - q_-\cM^-_{A\to B}$ be two HPTS maps, where $p_\pm$ and $q_\pm$ are non-negative real numbers, and $\cN^\pm_{A\to B}$ and $\cM^\pm_{A\to B}$ are quantum channels. By the definition of the diamond norm, the diamond distance between these two maps is
\begin{align}
    \frac{1}{2} \lrV{\widetilde{\cN}_{A\to B} - \widetilde{\cM}_{A\to B}}_\diamond &= \sup_{\psi_{RA}} \frac{1}{2} \lrV{\widetilde{\cN}_{A\to B}(\psi_{RA}) - \widetilde{\cM}_{A\to B}(\psi_{RA})}_1\\
    &= \sup_{\psi_{RA}} \bigg\{ \sup_{M: 0 \leq M \leq \idop} \tr\left[ M \left( \widetilde{\cN}_{A\to B}(\psi_{RA}) - \widetilde{\cM}_{A\to B}(\psi_{RA}) \right) \right] \nonumber\\
    &\qquad - \frac{1}{2}\tr\left[ \widetilde{\cN}_{A\to B}(\psi_{RA}) - \widetilde{\cM}_{A\to B}(\psi_{RA}) \right]  \bigg\},
\end{align}
where $\psi_{RA}$ is a pure state with $d_R = d_A$, and the second equality follows from the Helstrom-Holevo theorem (see, for example, Theorem 3.13 in Ref.~\cite{khatri2020principles}). Defining $p \coloneqq p_+ - p_-$ and $q \coloneqq q_+ - q_-$, we have
\begin{align}
    \frac{1}{2} \lrV{\widetilde{\cN}_{A\to B} - \widetilde{\cM}_{A\to B}}_\diamond = \sup_{\substack{\psi_{RA}\\ M: 0 \leq M \leq \idop}} \tr\left[ M \left( \widetilde{\cN}_{A\to B}(\psi_{RA}) - \widetilde{\cM}_{A\to B}(\psi_{RA}) \right) \right] - \frac{p - q}{2}.
\end{align}
Following the proof from Sec. 3.C.2 in Ref.~\cite{khatri2020principles}, it is easy to show that the first term on the right hand side of the above equation can be computed using the standard SDP for the diamond distance between two quantum channels. Hence, the diamond distance between two HPTS maps can be calculated as the result obtained from the SDP for two quantum channels minus the normalized difference between the trace scalars of the two maps, {\em i.e.},
\begin{align}
    \frac{1}{2} \lrV{\widetilde{\cN}_{A\to B} - \widetilde{\cM}_{A\to B}}_\diamond = \inf_{Z_{AB} \geq 0} \left\{ \mu - \frac{p-q}{2} ~\middle\vert~ Z_{AB} \geq J^{\widetilde{\cN}}_{AB} - J^{\widetilde{\cM}}_{AB},~ \tr_B[Z_{AB}] \leq \mu\idop_A \right\},
\end{align}
where $J^{\widetilde{\cN}}_{AB}$ and $J^{\widetilde{\cM}}_{AB}$ are the Choi operators of $\widetilde{\cN}_{A\to B}$ and $\widetilde{\cM}_{A\to B}$, respectively. Then, the minimum error $\varepsilon^*_{\gamma, \rm NS}\lrp{\cN, \cM}$ and minimum sampling cost $\gamma^*_{\varepsilon, \rm NS}$ can be written as SDPs in terms of the relevant maps' Choi operators.

\begin{proposition}\label{prop:general-sdp}
    Consider two quantum channels $\cN_{A\to B}$ and $\cM_{A'\to B'}$ whose Choi operators are $J^\cN_{AB}$ and $J^\cM_{A'B'}$, respectively.
    The minimum error of shadow simulation from $\cN$ to $\cM$ under no-signaling codes with a cost budget $\gamma$ is given by the following SDP:
    \begin{subequations}\label{appeq:general-sdp}
    \begin{align}
        \varepsilon^*_{\gamma, \rm NS}(\cN,\cM) = \inf &\; \varepsilon \\
        \text{\rm s.t.} &\; J^{\widetilde{\cM}}_{A'B'} = \tr_{AB}\Big[ \lrp{\lrp{J^\cN_{AB}}^T \ox \idop_{A'B'}} \lrp{J^{\widehat\cS^+}_{A'BAB'} - J^{\widehat\cS^-}_{A'BAB'}} \Big], \label{appeq:general-sdp_simulated-map}\\
        &\; Z_{A'B'} \geq 0,~ Z_{A'B'} \geq J^\cM_{A'B'} - J^{\widetilde{\cM}}_{A'B'},~ \tr_{B'}[Z_{A'B'}] \leq \frac{2\varepsilon + 1 - p_+ + p_-}{2}\idop_{A'}, \\
        &\; J^{\widehat\cS^\pm}_{A'BAB'} \geq 0,~ \tr_{AB'}\lrb{J^{\widehat\cS^\pm}_{A'BAB'}} = p_\pm \idop_{A'B},~ p_+ + p_- \leq \gamma, \label{appeq:general-sdp_cptp}\\
        &\; \tr_{B'}\lrb{J^{\widehat\cS^\pm}_{A'BAB'}} = \frac{1}{d_B}\tr_{BB'}\lrb{J^{\widehat\cS^\pm}_{A'BAB'}} \ox \idop_B, \label{appeq:general-sdp_no-sig_1}\\
        &\; \tr_{A}\lrb{J^{\widehat\cS^\pm}_{A'BAB'}} = \frac{1}{d_{A'}}\tr_{A'A}\lrb{J^{\widehat\cS^\pm}_{A'BAB'}} \ox \idop_{A'}. \label{appeq:general-sdp_no-sig_2}
    \end{align}
    \end{subequations}
    Similarly, the minimum sampling cost $\gamma^*_{\varepsilon, \rm NS}(\cN,\cM)$ under an error tolerance $\varepsilon$ is given by changing the optimization objective of the above SDP to $p_+ + p_-$ and removing the condition $p_+ + p_- \leq \gamma$.
\end{proposition}

%%%%%%%%%%%%%%%%%%%%%%%%%%%%%%%%%%%%%%%%%
%%%%%%%%%%%%%%%%%%%%%%%%%%%%%%%%%%%%%%%%%
\section{Appendix C: Shadow Communication}
\renewcommand{\theequation}{C\arabic{equation}}
\setcounter{equation}{0}
\renewcommand{\theHequation}{C\arabic{equation}}
In this section, we derive the SDP for $Q^{(1)}_{\gamma, \rm NS}$, the one-shot zero-error $\gamma$-cost shadow capacity assisted by no-signaling codes, given in Theorem~\ref{thm:zero_error_capacity}. To achieve this, we first need to derive SDPs for some other quantities, which are of interest on their own.

First, we tailor the general SDPs of minimum error and minimum sampling cost for the shadow communication task. The original SDPs are given in Proposition~\ref{prop:general-sdp}. The target channel $\cM_{A'\to B'}$ becomes $\cI_d$, where $d=d_{A'}=d_{B'}$ is the dimension of the target noiseless channel.

\begin{lemma}\label{lemma:shadow_comm_min_err}
    Given a fixed dimension $d = d_{A'} = d_{B'}$ and an error tolerance $\varepsilon$, the minimum error of shadow simulation from $\cN_{A\to B}$ to $\cI_d$ under no-signaling codes with a cost budget $\gamma$ is given by the following SDP:
    \begin{subequations}\label{appeq:comm-sdp-min-err}
    \begin{align}
        \varepsilon^*_{\gamma, \rm NS} (\cN,\cI_d) = \inf &\; \varepsilon \\
        \text{\rm s.t.} &\; J^{\widetilde{\cM}}_{A'B'} = \tr\lrb{\lrp{J^\cN_{AB}}^T \lrp{T^+_{AB} - T^-_{AB}}} \frac{d^2\Phi_{A'B'} - \idop_{A'B'}}{d\lrp{d^2 - 1}} + \tr\lrb{V^+_A - V^-_A} \frac{\idop_{A'B'} - \Phi_{A'B'}}{d\lrp{d^2 - 1}}, \label{appeq:comm-sdp-min-err_simulated-map}\\
        &\; Z_{A'B'} \geq 0,~ Z_{A'B'} \geq J^{\cI_d}_{A'B'} - J^{\widetilde{\cM}}_{A'B'},~ \tr_{B'}[Z_{A'B'}] \leq \frac{1}{2}\lrp{2\varepsilon + 1 - \frac{\tr\lrb{V^+_A - V^-_A}}{d^2}}\idop_{A'}, \\
        &\; 0 \leq T^\pm_{AB} \leq V^\pm_A \ox \idop_B,~ \tr_A\lrb{T^\pm_{AB}} = \frac{\tr\lrb{V^\pm_A}}{d^2} \idop_B, \frac{\tr\lrb{V^+_A + V^-_A}}{d^2} \leq \gamma.
    \end{align}
    \end{subequations}
    Similarly, the minimum sampling cost $\gamma^*_{\varepsilon, \rm NS} (\cN,\cI_d)$ with an error tolerance $\varepsilon$ is given by changing the optimization objective of the above SDP to $\tr\lrb{V^+_A + V^-_A}/d^2$ and removing the condition $\tr\lrb{V^+_A + V^-_A} / d^2 \leq \gamma$.
\end{lemma}
\begin{proof}
When the target identity channel has dimension $d$, and we denote it by $\cI_d$, the minimum error and the minimum sampling cost of shadow communication over the channel $\cN$ are $\varepsilon^*_{\gamma, \rm NS} (\cN,\cI_d)$ and $\gamma^*_{\varepsilon, \rm NS} (\cN,\cI_d)$, respectively, where $\varepsilon$ and $\gamma$ are the error tolerance and cost budget.
Below, we exploit the symmetry of optimal solutions under twirling to obtain simplified SDPs for both quantities.

Consider the SDP for minimum error in Proposition~\ref{prop:general-sdp} with the target channel being $\cI_d$ first. Note that if $J^{\widehat\cS^\pm}_{A'BAB'}$ are optimal, then for any $d$-dimensional unitary $U$, the Choi operators
\begin{align}
    \bar{J}^{\widehat\cS^\pm}_{A'BAB'} \coloneqq (U_{A'}\ox \overline{U}_{B'}) J^{\widehat\cS^\pm}_{A'BAB'} (U_{A'}\ox \overline{U}_{B'})^\dagger
\end{align}
are also optimal, where $\overline{U}$ denotes the complex conjugate of $U$. The optimality of $\bar{J}^{\widehat\cS^\pm}_{A'BAB'}$ can be checked by verifying that they satisfy all the conditions in the original SDP while keeping the value of $\varepsilon$ unchanged.
Due to the linearity of the constraints, any convex combination of optimal Choi operators is still optimal. Hence, we now redefine
\begin{align}
    \bar{J}^{\widehat\cS^\pm}_{A'BAB'} \coloneqq \int dU (U_{A'}\ox \overline{U}_{B'}) J^{\widehat\cS^\pm}_{A'BAB'} (U_{A'}\ox \overline{U}_{B'})^\dagger,
\end{align}
where the integral is taken over the Haar measure on the unitary group. This new pair of Choi operators are also optimal. It was shown in Ref.~\cite{rains2001semidefinite} that the twirling operation $\cT_{A'B'}: X_{A'B'} \mapsto \int dU (U_{A'}\ox \overline{U}_{B'}) X_{A'B'} (U_{A'}\ox \overline{U}_{B'})^\dagger$ has the following action:
\begin{align}
    \cT_{A'B'}(X_{A'B'}) = \tr[X_{A'B'} \Phi_{A'B'}] \Phi_{A'B'} + \frac{\tr[X_{A'B'} (\idop_{A'B'} - \Phi_{A'B'})]}{d^2 - 1} (\idop_{A'B'} - \Phi_{A'B'}),
\end{align}
where $\Phi_{A'B'} = \proj{\Phi}_{A'B'}$ is the maximally entangled state with $\ket{\Phi}_{A'B'} \coloneqq \frac{1}{\sqrt{d}} \sum_{j=0}^{d-1} \ket{j}_{A'}\ket{j}_{B'}$.
Thus,
\begin{align}
    \bar{J}^{\widehat\cS^\pm}_{A'BAB'} = \Phi_{A'B'} \ox \tr_{A'B'}\lrb{J^{\widehat\cS^\pm}_{A'BAB'} \Phi_{A'B'}} + (\idop_{A'B'} - \Phi_{A'B'}) \ox \frac{\tr_{A'B'}\lrb{J^{\widehat\cS^\pm}_{A'BAB'} (\idop_{A'B'} - \Phi_{A'B'})}}{d^2 - 1}.
\end{align}
Without any constraints, $\tr_{A'B'}\lrb{J^{\widehat\cS^\pm}_{A'BAB'} \Phi_{A'B'}}$ and $\tr_{A'B'}\lrb{J^{\widehat\cS^\pm}_{A'BAB'} (\idop_{A'B'} - \Phi_{A'B'})} / \lrp{d^2 - 1}$ can be any linear operators. We denote them by $T^\pm_{AB}$ and $W^\pm_{AB}$, respectively, so that
\begin{align}
    \bar{J}^{\widehat\cS^\pm}_{A'BAB'} &= \Phi_{A'B'} \ox T^\pm_{AB} + (\idop_{A'B'} - \Phi_{A'B'}) \ox W^\pm_{AB}.
\end{align}

We now express SDP~\eqref{appeq:general-sdp} in terms of $T^\pm_{AB}$ and $W^\pm_{AB}$.
The Choi operator of the simulated map $\widetilde{\cM}$ in Eq.~\eqref{appeq:general-sdp_simulated-map} can be written as
\begin{align}
    J^{\widetilde{\cM}}_{A'B'} = \tr\lrb{(J^\cN_{AB})^T (T^+_{AB} - T^-_{AB})} \Phi_{A'B'} + \tr\lrb{(J^\cN_{AB})^T (W^+_{AB} - W^-_{AB})} (\idop_{A'B'} - \Phi_{A'B'}).
\end{align}
The first inequality in condition~\eqref{appeq:general-sdp_cptp} becomes $T^\pm_{AB} \geq 0$ and $W^\pm_{AB} \geq 0$, and the equality in condition~\eqref{appeq:general-sdp_cptp} can be written as
\begin{align}
    \frac{\idop_{A'}}{d} \ox \tr_{A}\lrb{T^\pm_{AB} + (d^2 - 1)W^\pm_{AB}} = p_\pm \idop_{A'B},
\end{align}
which is equivalent to the requirement that
\begin{align}\label{appeq:simplified_trace_scaling}
    \tr_{A}\lrb{T^\pm_{AB} + (d^2 - 1)W^\pm_{AB}} = d p_\pm \idop_{B}.
\end{align}
For the $B$-to-$A$ no-signaling condition~\eqref{appeq:general-sdp_no-sig_1}, its left-hand side can be written as
\begin{align}
    \tr_{B'}\lrb{J^{\widehat\cS^\pm}_{A'BAB'}} = \frac{\idop_{A'}}{d} \ox \lrp{T^\pm_{AB} + \lrp{d^2 - 1}W^\pm_{AB}},
\end{align}
and its right-hand side can be written as
\begin{align}
    \frac{1}{d_B} \tr_{BB'}\lrb{J^{\widehat\cS^\pm}_{A'BAB'}} \ox \idop_B = \frac{\idop_{A'}}{d} \ox \tr_B\lrb{T^\pm_{AB} + \lrp{d^2 - 1}W^\pm_{AB}} \ox \frac{\idop_B}{d_B}.
\end{align}
Hence, the condition~\eqref{appeq:general-sdp_no-sig_1} is equivalent to
\begin{align}
    T^\pm_{AB} + \lrp{d^2 - 1}W^\pm_{AB} = \tr_B\lrb{T^\pm_{AB} + \lrp{d^2 - 1}W^\pm_{AB}} \ox \frac{\idop_B}{d_B}.
\end{align}
Similarly, the right-hand side of the $A$-to-$B$ no-signaling condition~\eqref{appeq:general-sdp_no-sig_2} can be simplified as
\begin{align}
    \frac{\idop_{A'B'}}{d^2} \ox \tr_A\lrb{T^\pm_{AB} + \lrp{d^2 - 1}W^\pm_{AB}} = \frac{p_\pm\idop_{A'BB'}}{d}
\end{align}
due to Eq.~\eqref{appeq:simplified_trace_scaling}.
Hence, the condition~\eqref{appeq:general-sdp_no-sig_2} is equivalent to
\begin{align}
    \Phi_{A'B'} \ox \tr_A\lrb{T^\pm_{AB}} + (\idop_{A'B'} - \Phi_{A'B'}) \ox \tr_A\lrb{W^\pm_{AB}} = \frac{p_\pm\idop_{A'BB'}}{d}.
\end{align}
Note that the equation above holds if and only if $\tr_A\lrb{T^\pm_{AB}} = \tr_A\lrb{W^\pm_{AB}} = \frac{p_\pm \idop_B}{d}$, which implies Eq.~\eqref{appeq:simplified_trace_scaling}.

Now the original SDP has been simplified to
\begin{subequations}
\begin{align}
    \varepsilon^*_{\gamma, \rm NS} (\cN,\cI_d) = \inf &\; \varepsilon \\
    \text{s.t.} &\; J^{\widetilde{\cM}}_{A'B'} = \tr\lrb{\lrp{J^\cN_{AB}}^T \lrp{T^+_{AB} - T^-_{AB}}} \Phi_{A'B'} + \tr\lrb{\lrp{J^\cN_{AB}}^T \lrp{W^+_{AB} - W^-_{AB}}} (\idop_{A'B'} - \Phi_{A'B'}), \\
    &\; Z_{A'B'} \geq 0,~ Z_{A'B'} \geq J^{\cI_d}_{A'B'} - J^{\widetilde\cM}_{A'B'},~ \tr_{B'}[Z_{A'B'}] \leq \frac{2\varepsilon + 1 - p_+ + p_-}{2}\idop_{A'}, \\
    &\; T^\pm_{AB} \geq 0,~ W^\pm_{AB} \geq 0,~ \tr_A\lrb{T^\pm_{AB}} = \tr_A\lrb{W^\pm_{AB}} = \frac{p_\pm \idop_B}{d},~ p_+ + p_- \leq \gamma, \\
    &\; T^\pm_{AB} + \lrp{d^2 - 1}W^\pm_{AB} = \tr_B\lrb{T^\pm_{AB} + \lrp{d^2 - 1}W^\pm_{AB}} \ox \frac{\idop_B}{d_B}. \label{appeq:sdp_lower_sample_simp_t-and-w}
\end{align}
\end{subequations}
Denoting $V^\pm_{A} \coloneqq d \tr_B\lrb{T^\pm_{AB} + \lrp{d^2 - 1}W^\pm_{AB}}/d_B$, by condition~\eqref{appeq:sdp_lower_sample_simp_t-and-w}, we can write the variables $W^\pm_{AB}$ in terms of $V^\pm_A$ and $T^\pm_{AB}$ as
\begin{align}
    \lrp{d^2 - 1}dW^\pm_{AB} = V^\pm_A \ox \idop_B - dT^\pm_{AB}.
\end{align}
Then, other conditions involving $W^\pm_{AB}$ can also be written as conditions on $V^\pm_A$ and $T^\pm_{AB}$. The condition $W^\pm_{AB} \geq 0$ becomes $V^\pm_A \ox \idop_B \geq d T^\pm_{AB}$. The condition $\tr_A\lrb{W^\pm_{AB}} = \frac{p_\pm \idop_B}{d}$ becomes
\begin{align}
    d \tr_A\lrb{W^\pm_{AB}} = \frac{\tr\lrb{V^\pm_A}\idop_B - p_\pm\idop_B}{d^2 - 1} = p_\pm \idop_B,
\end{align}
which is equivalent to requiring $\tr\lrb{V^\pm_A} = d^2 p_\pm$.
Finally, the Choi operator of the simulated map $\widetilde{\cM}$ can be written as
\begin{align}
    J^{\widetilde{\cM}}_{A'B'} &= \tr\lrb{\lrp{J^\cN_{AB}}^T \Delta T_{AB}} \Phi_{A'B'} + \tr\lrb{\lrp{J^\cN_{AB}}^T \left(\frac{\Delta V_A \ox \idop_B - d\Delta T_{AB}}{d\lrp{d^2 - 1}}\right)} (\idop_{A'B'} - \Phi_{A'B'})\\
    &= \tr\lrb{\lrp{J^\cN_{AB}}^T \Delta T_{AB}} \frac{d^2\Phi_{A'B'} - \idop_{A'B'}}{d^2 - 1} + \tr\lrb{\lrp{J^\cN_A}^T \Delta V_A} \frac{\idop_{A'B'} - \Phi_{A'B'}}{d\lrp{d^2 - 1}},
\end{align}
where we denote $\Delta T_{AB} \coloneqq T^+_{AB} - T^-_{AB}$, $\Delta V_A \coloneqq V^+_A - V^-_A$, and $J^\cN_A \coloneqq \tr_B\lrb{J^\cN_{AB}}$. Because $\cN_{A\to B}$ is a quantum channel, the partial trace of its Choi operator over the output system $B$ equals $\idop_{A}$. Hence,
\begin{align}
    J^{\widetilde{\cM}}_{A'B'} &= \tr\lrb{\lrp{J^\cN_{AB}}^T \Delta T_{AB}} \frac{d^2\Phi_{A'B'} - \idop_{A'B'}}{d^2 - 1} + \tr\lrb{\Delta V_A} \frac{\idop_{A'B'} - \Phi_{A'B'}}{d\lrp{d^2 - 1}}.
\end{align}
By further relabeling $dT^\pm_{AB}$ as $T^\pm_{AB}$ and replacing $p_\pm$ with $\tr\lrb{V^\pm_A} / d^2$ lead to SDP~\eqref{appeq:comm-sdp-min-err}.

Because $p_+ + p_- = \tr\lrb{V^+_A + V^-_A} / d^2$, by Proposition~\ref{prop:general-sdp}, we know changing the optimization objective of SDP~\eqref{appeq:comm-sdp-min-err} to $\tr\lrb{V^+_A + V^-_A} / d^2$ and removing the condition $\tr\lrb{V^+_A + V^-_A} / d^2 \leq \gamma$ gives us an SDP for $\gamma^*_{\varepsilon, \rm NS} (\cN,\cI_d)$.
\end{proof}

%%%%%%%%%%%%%%%%%%%%%%%%%%%%%%%%%%%%%%%%%
Now, we turn to zero-error shadow communication. In this case, the preset error tolerance $\varepsilon$ is $0$. We can greatly simplify the SDP for $\gamma^*_{\varepsilon, \rm NS}(\cN,\cI_d)$ using the fact $\varepsilon = 0$.

\begin{lemma}\label{lemma:communication_zero-error_sample-complexity}
    The zero-error minimum sampling cost of shadow communication with ${\rm NS}$ codes is given by the following SDP:
    \begin{subequations}\label{appeq:communication_zero-error_sample-complexity}
    \begin{align}
        \gamma^*_{0, \rm NS} (\cN,\cI_d) = \inf &\; \frac{\tr\lrb{V^+_A + V^-_A}}{d^2} \\
        \text{\rm s.t.} &\; \tr\lrb{\lrp{J^\cN_{AB}}^T \lrp{T^+_{AB} - T^-_{AB}}} = \tr\lrb{V^+_A - V^-_A} = d^2, \\
        &\; 0 \leq T^\pm_{AB} \leq V^\pm_A \ox \idop_B,~ \tr_A\lrb{T^\pm_{AB}} = \frac{\tr\lrb{V^\pm_A}}{d^2} \idop_B.
    \end{align}
    \end{subequations}
\end{lemma}
\begin{proof}
Consider the SDP for $\gamma^*_{\varepsilon, \rm NS}(\cN,\cI_d)$ given in Lemma~\ref{lemma:shadow_comm_min_err}.
For $\varepsilon = 0$, we have
\begin{align}
    \tr_{B'}[Z_{A'B'}] \leq \frac{1}{2} \lrp{1 - \frac{\tr\lrb{\Delta V_A}}{d^2}} \idop_{A'},
\end{align}
where $\Delta V_A \coloneqq V^+_A - V^-_A$.
On the other hand, taking the partial trace of $J^{\widetilde{\cM}}_{A'B'}$ over system $B'$, we get
\begin{align}
    \tr_{B'}\lrb{J^{\widetilde{\cM}}_{A'B'}} = \tr\lrb{\lrp{J^\cN_{AB}}^T \Delta T_{AB}} \frac{d\idop_{A'} - d\idop_{A'}}{d\lrp{d^2 - 1}} + \tr\lrb{\Delta V_A} \frac{d\idop_{A'} - \idop_{A'}/d}{d\lrp{d^2 - 1}} = \tr\lrb{\Delta V_A} \frac{\idop_{A'}}{d^2}.
\end{align}
Then, it follows that
\begin{align}
    \frac{1}{2} \lrp{1 - \frac{\tr\lrb{\Delta V_A}}{d^2}} \idop_{A'} \geq \tr_{B'}[Z_{A'B'}] \geq \tr_{B'}\lrb{J^{\cI_d}_{A'B'} - J^{\widetilde{\cM}}_{A'B'}} = \lrp{1 - \frac{\tr\lrb{\Delta V_A}}{d^2}} \idop_{A'}.
\end{align}
For $Z_{A'B'} \geq 0$, we conclude that $Z_{A'B'} = 0$ and $\tr\lrb{\Delta V_A} = d^2$, implying $J^{\widetilde{\cM}}_{A'B'} = J^\cI_{A'B'}$.

Note that we can also write $J^{\widetilde{\cM}}_{A'B'}$ as
\begin{align}
    J^{\widetilde{\cM}}_{A'B'} = \tr\lrb{\lrp{J^\cN_{AB}}^T \frac{\Delta T_{AB}}{d}} \Phi_{A'B'} + \tr\lrb{\lrp{J^\cN_{AB}}^T \left(\frac{\Delta V_A \ox \idop_B - \Delta T_{AB}}{d\lrp{d^2 - 1}}\right)} (\idop_{A'B'} - \Phi_{A'B'})
\end{align}
by reorganizing Eq.~\eqref{appeq:comm-sdp-min-err_simulated-map} with $\Delta T_{AB} \coloneqq T^+_{AB} - T^-_{AB}$. Because $J^{\widetilde{\cM}}_{A'B'} = J^\cI_{A'B'} = d\Phi_{A'B'}$, it must be true that
\begin{align}
    \tr\lrb{\lrp{J^\cN_{AB}}^T \Delta T_{AB}} = \tr\lrb{\Delta V_A} = d^2.
\end{align}
Hence the proof.
\end{proof}

%%%%%%%%%%%%%%%%%%%%%%%%%%%%%%%%%%%%%%%%%
From this lemma, we can derive an SDP for the one-shot zero-error $\gamma$-cost shadow capacity as follows.

\begin{theorem}\label{thm:zero_error_capacity}
    The one-shot zero-error $\gamma$-cost shadow capacity assisted by no-signaling codes of a quantum channel $\cN_{A\to B}$ is given by the following SDP:
    \begin{subequations}\label{appeq:zero_error_capacity_sdp}
    \begin{align}
        Q^{(1)}_{\gamma, \rm NS} (\cN) = \sup &\; \log_2~ \left\lfloor \sqrt{\tr\lrb{V_A}} \right\rfloor \\
        \text{\rm s.t.} &\; \tr\lrb{\lrp{J^\cN_{AB}}^T T_{AB}} = \tr\lrb{V_A},~ \tr_A\lrb{T_{AB}} = \idop_B,~ 0 \leq T_{AB} + R_{AB} \leq (V_A + W_A) \ox \idop_B, \\
        &\; 0 \leq R_{AB} \leq W_A \ox \idop_B,~ \tr_A\lrb{R_{AB}} = \frac{\gamma - 1}{2} \idop_B,~ \tr\lrb{W_A} = \frac{\gamma - 1}{2} \tr\lrb{V_A}.
    \end{align}
    \end{subequations}
\end{theorem}
\begin{proof}
The SDP given in Lemma~\ref{lemma:communication_zero-error_sample-complexity} allows us to formulate $Q^{(1)}_{\gamma, \rm NS}(\cN)$ as an optimization problem by replacing $d^2$ with $\tr\lrb{V^+_A - V^-_A}$, and the objective of the optimization is to maximize $\log_2~ \left\lfloor \sqrt{\tr\lrb{V^+_A - V^-_A}} \right\rfloor$ according to the definition of $Q^{(1)}_{\gamma, \rm NS}$:
\begin{subequations}
\begin{align}
    Q^{(1)}_{\gamma, \rm NS}(\cN) = \sup &\; \log_2~ \left\lfloor \sqrt{\tr\lrb{V^+_A - V^-_A}} \right\rfloor\\
    \text{\rm s.t.} &\; \tr\lrb{\lrp{J^\cN_{AB}}^T \lrp{T^+_{AB} - T^-_{AB}}} = \tr\lrb{V^+_A - V^-_A},~ \frac{\tr\lrb{V^+_A + V^-_A}}{\tr\lrb{V^+_A - V^-_A}} \leq \gamma, \label{appeq:ns-capacity_original-sdp_inequality}\\
    &\; 0 \leq T^\pm_{AB} \leq V^\pm_A \ox \idop_B,~ \tr_A\lrb{T^\pm_{AB}} = \frac{\tr\lrb{V^\pm_A} \idop_B}{\tr\lrb{V^+_A - V^-_A}}, \label{appeq:ns-capacity_original-sdp_trace-a-t}
\end{align}
\end{subequations}
where the inequality in condition~\eqref{appeq:ns-capacity_original-sdp_inequality} corresponds to the limited cost budget.
This is not an SDP, but observe that the equality in condition~\eqref{appeq:ns-capacity_original-sdp_trace-a-t} is equivalent to the following two equations:
\begin{align}
    \tr_A\lrb{T^+_{AB} - T^-_{AB}} = \idop_B \quad\text{and}\quad \tr_A\lrb{T^+_{AB} + T^-_{AB}} = \frac{\tr\lrb{V^+_A + V^-_A}}{\tr\lrb{V^+_A - V^-_A}}\idop_B.
\end{align}
In addition, the inequality in condition~\eqref{appeq:ns-capacity_original-sdp_inequality} can be restricted to equality without affecting the optimization result because if $\bar{T}^\pm_{AB}$ and $\bar{V}^\pm_{A}$ form a set of optimal solution such that $\tr\lrb{\bar{V}^+_A + \bar{V}^-_A} / \tr\lrb{\bar{V}^+_A - \bar{V}^-_A} = \bar{\gamma} < \gamma$, then $\bar{T}'^\pm_{AB} \coloneqq \bar{T}^\pm_{AB} + (\gamma - \bar{\gamma}) \idop_{AB} / 2d_A$ and $\bar{V}'^\pm_{A} \coloneqq \bar{V}^\pm_{A} + \lrp{\gamma - \bar{\gamma}} \tr\lrb{\bar{V}^+_A - \bar{V}^-_A} \idop_{A} / 2d_A$ also form a set of optimal solution with $\tr\lrb{\bar{V}'^+_A + \bar{V}'^-_A} / \tr\lrb{\bar{V}'^+_A - \bar{V}'^-_A} = \gamma$.
Note that for $\bar{T}'^\pm_{AB}$ and $\bar{V}'^\pm_A$ to be valid solution, it must be true that $\tr\lrb{\bar{V}^+_A - \bar{V}^-_A} \geq 1$ so that $\bar{V}'^\pm_A \ox \idop_B \geq \bar{T}'^\pm_{AB}$. This is indeed the case because from the constraint $T^\pm_{A'B'} \leq V^\pm_A \ox \idop_B$ we have $\tr_A\lrb{\bar{T}^+_{AB} + \bar{T}^-_{AB}} \leq \tr\lrb{\bar{V}^+_A + \bar{V}^-_A} \idop_B$. For the equality $\tr_A\lrb{\bar{T}^+_{AB} + \bar{T}^-_{AB}} = \tr\lrb{\bar{V}^+_A + \bar{V}^-_A} \idop_B / \tr\lrb{\bar{V}^+_A - \bar{V}^-_A}$ to hold, $\tr\lrb{\bar{V}^+_A - \bar{V}^-_A}$ has to be larger than or equal to $1$.

Now we can safely require $\tr\lrb{V^+_A + V^-_A} = \gamma \tr\lrb{V^+_A - V^-_A}$ and thus $\tr_A\lrb{T^+_{AB} + T^-_{AB}} = \gamma \idop_B$. Changing the variables to $T_{AB} \coloneqq T^+_{AB} - T^-_{AB}$, $R_{AB} \coloneqq T^-_{AB}$, $V_A \coloneqq V^+_A - V^-_A$, and $W_A \coloneqq V^-_A$ results in the claimed SDP.
\end{proof}

In the main text, we claimed that $Q^{(1)}_{\gamma, \rm NS}$ generalizes the no-signaling-assisted one-shot zero-error quantum capacity $Q^{(1)}_{\rm NS}$ in the sense that $Q^{(1)}_{1, \rm NS}(\cN) = Q^{(1)}_{\rm NS}(\cN)$ for any quantum channel $\cN$.
To see this, note that when $\gamma = 1$, variables $R_{AB}$ and $W_A$ from SDP~\eqref{appeq:zero_error_capacity_sdp} satisfy $\tr_A\lrb{R_{AB}} = 0$ and $\tr\lrb{W_A} = 0$. Because both $R_{AB}$ and $W_A$ are positive semidefinite operators, they can only be $0$. Therefore, the original SDP~\eqref{appeq:zero_error_capacity_sdp} reduces to
\begin{subequations}
\begin{align}
    \sup &\; \log_2~ \left\lfloor \sqrt{\tr\lrb{V_A}} \right\rfloor \\
    \text{\rm s.t.} &\; \tr\lrb{\lrp{J^\cN_{AB}}^T T_{AB}} = \tr\lrb{V_A},~ \tr_A\lrb{T_{AB}} = \idop_B,~ 0 \leq T_{AB} \leq V_A \ox \idop_B,
\end{align}
\end{subequations}
which is an SDP for $Q^{(1)}_{\rm NS}(\cN)$~\cite{duan2015no, wang2016semidefinite}.

Below, we provide an exact characterization of $Q^{(1)}_{\gamma, \rm NS}$ for single-qubit depolarizing channels.

\renewcommand\theproposition{\ref{thm:depo_capacity}}
\setcounter{proposition}{\arabic{proposition}-1}
\begin{theorem}
    Let  $\cN_{{\rm depo}, p}(\rho) = p\rho + (1-p)\idop_2/2$ be a single-qubit depolarizing channel, where $p\in [0,1]$ is a probability and $\idop_2$ is the identity operator on $\mathbb{C}^2$.
    For $\gamma \geq 1$, the one-shot zero-error shadow capacity assisted by no-signaling resources is
    \begin{align}
        Q^{(1)}_{\gamma, \rm NS}(\cN_{{\rm depo}, p}) = \log_2 \left\lfloor \sqrt{2p\gamma + p + 1} \right\rfloor.
    \end{align}
\end{theorem}
\renewcommand{\theproposition}{\arabic{proposition}}
\begin{proof}
The Choi operator of the depolarizing channel $\cN_{{\rm depo}, p}$ from qubit system $A$ to qubit system $B$ is
\begin{align}
    J^{\cN_{{\rm depo}, p}}_{AB} = p\Gamma_{AB} + \frac{1-p}{2}\idop_{AB},
\end{align}
where $\Gamma_{AB} \coloneqq \proj{\Gamma}_{AB}$ is the unnormalized maximally entangled state with $\ket{\Gamma}_{AB} \coloneqq \sum_{j=0}^1 \ket{j}_{A}\ket{j}_{B}$.
We first consider the case where $p \in (0,1]$.
It is straightforward to verify that
\begin{align}
    \bar{T}^+_{AB} \coloneqq \frac{d^2-1+p}{4p} \Gamma_{AB},&\quad \bar{T}^-_{AB} \coloneqq \frac{d^2-1-3p}{6p} \idop_{AB} - \frac{d^2-1-3p}{12p} \Gamma_{AB},\\
    \bar{V}^+_{A} \coloneqq \frac{d^2(d^2 - 1 + p)}{8p} \idop_A,&\quad \bar{V}^-_{A} \coloneqq \frac{d^2(d^2-1-3p)}{8p} \idop_A
\end{align}
form a feasible solution to the SDP for $\gamma^*_{0, \rm NS}(\cN_{{\rm depo}, p}, \cI_d)$ as presented in Lemma~\ref{lemma:communication_zero-error_sample-complexity}, implying
\begin{align}\label{appeq:depo_sample_upper}
    \gamma^*_{0, \rm NS}(\cN_{{\rm depo}, p}, \cI_d) \leq \frac{\tr\lrb{\bar{V}^+_{A} + \bar{V}^-_{A}}}{d^2} = \frac{d^2-1-p}{2p}.
\end{align}
Using the Lagrange dual function, we can derive that the following problem is the dual problem associated with SDP~\eqref{appeq:communication_zero-error_sample-complexity}:
\begin{align}
    \gamma^*_{0, \rm NS}(\cN, \cI_d) \geq \sup &\; \lambda - \mu \\
    \text{\rm s.t.} &\; M^\pm_{AB} \geq 0,~ M^\pm_{AB} + d^2\idop_A \ox N^\pm_{B} \geq \pm\lambda\lrp{J^\cN_{AB}}^T, \\
    &\; \tr_B\lrb{M^\pm_{AB}} = \lrp{1 - \tr\lrb{N^\pm_B} \pm \mu}\idop_A.
\end{align}
Again, it is straightforward to verify that, given $J^\cN_{AB} = J^{\cN_{{\rm depo}, p}}_{AB}$,
\begin{align}
    \bar{\lambda} \coloneqq \frac{d^2}{2p},\quad \bar{\mu} \coloneqq \frac{1+p}{2p},\quad \bar{N}^+_{B} \coloneqq \frac{1+3p}{4p}\idop_B,\quad \bar{N}^-_{B} \coloneqq \frac{p-1}{4p}\idop_B,\quad M^\pm_{AB} = 0
\end{align}
form a feasible solution to the dual problem, implying
\begin{align}\label{appeq:depo_sample_lower}
    \gamma^*_{0, \rm NS}(\cN_{{\rm depo}, p}, \cI_d) \geq \bar{\lambda} - \bar{\mu} = \frac{d^2-1-p}{2p}.
\end{align}
Combining Eqs.~\eqref{appeq:depo_sample_upper} and~\eqref{appeq:depo_sample_lower}, we conclude that
\begin{align}
    \gamma^*_{0, \rm NS}(\cN_{{\rm depo}, p}, \cI_d) = \frac{d^2-1-p}{2p},
\end{align}
which is the minimum sampling cost required to simulate the $d$-dimensional identity channel $\cI_d$ from $\cN_{{\rm depo}, p}$. In other words, for any $\gamma$ such that
\begin{align}
    \frac{d^2-1-p}{2p} \leq \gamma < \frac{(d+1)^2-1-p}{2p},
\end{align}
we have $Q^{(1)}_{\gamma, \rm NS}(\cN_{{\rm depo}, p}) = \log_2 d$. Solving for the value of $d$ in terms of $\gamma$, we obtain $d = \left\lfloor\sqrt{2p\gamma + p + 1}\right\rfloor$, and hence
\begin{align}
    Q^{(1)}_{\gamma, \rm NS}(\cN_{{\rm depo}, p}) = \log_2 \left\lfloor\sqrt{2p\gamma + p + 1}\right\rfloor.
\end{align}

When $p=0$, the Choi operator of the depolarizing channel is simply $\idop_{AB} / 2$. Taking this into SDP~\eqref{appeq:zero_error_capacity_sdp}, we see that $\tr[V_A]$ can only takes a fixed value of $1$. Hence, $Q^{(1)}_{\gamma, \rm NS}(\cN_{{\rm depo}, p}) = 0$ for $p=0$ and arbitrary $\gamma$, which coincides with the value that $\log_2 \left\lfloor\sqrt{2p\gamma + p + 1}\right\rfloor$ evaluates to. Therefore, $Q^{(1)}_{\gamma, \rm NS}(\cN_{{\rm depo}, p}) = \log_2 \left\lfloor\sqrt{2p\gamma + p + 1}\right\rfloor$ for any $p\in[0,1]$.
\end{proof}

To further showcase the difference between shadow simulation and conventional quantum channel simulation, we, in addition to the single-qubit depolarizing channel $\cN_{\rm depo}(\cdot) = p(\rho) + (1-p)\idop_2/2$, consider two other common quantum channels:
the single-qubit amplitude damping channel $\cN_{\rm AD}$ with two Kraus operators $\proj{0} + \sqrt{p}\proj{1}$ and $\sqrt{1-p}\ketbra{0}{1}$ and the single-qubit dephasing channel $\cN_{\rm deph}(\cdot) = p(\cdot) + (1-p){\rm diag}(\cdot)$.
For each channel, the parameter $p \in [0,1]$ indicates the level of noise.

\begin{figure}[h]
    \centering
    \includegraphics[width=0.49\textwidth]{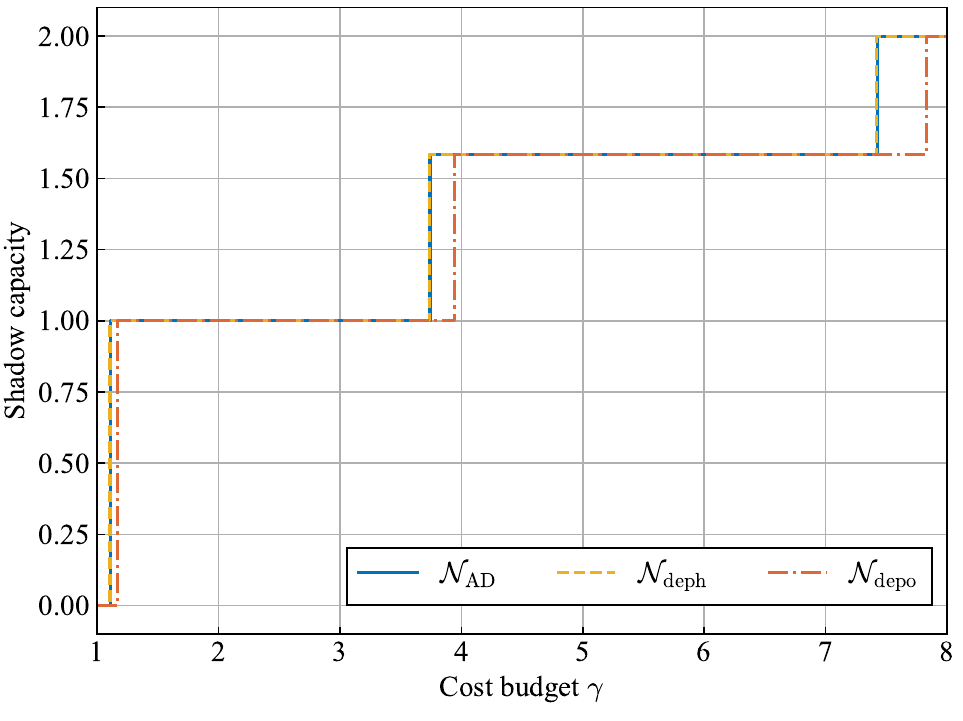}
    \caption{{\bf Comparison between the conventional no-signaling-assisted quantum communication and no-signaling-assisted shadow communication under different cost budgets.}
    Compared with the quantum case, where the one-shot zero-error quantum capacities are $0$ for all channels, higher one-shot zero-error shadow capacity is achieved for every channel with increased cost budget.}
    \label{fig:capacity}
\end{figure}

In Fig.~\ref{fig:capacity}, we present some numerical results on these channels at a low noise level ($p=0.9$).
For conventional quantum communication with no-signaling codes, all these channels' one-shot zero-error capacities are zero. For shadow communication, on the other hand, the one-shot zero-error capacities of these three channels steadily go up as the budget for sampling cost increases. The stepwise changes in the zero-error capacity show that we can trade in computational resources for better performance in shadow communication, attaining computational power beyond purely quantum protocols.

%%%%%%%%%%%%%%%%%%%%%%%%%%%%%%%%%%%%%%%%%
%%%%%%%%%%%%%%%%%%%%%%%%%%%%%%%%%%%%%%%%%
\section{Appendix D: Channel Formation}
\renewcommand{\theequation}{D\arabic{equation}}
\setcounter{equation}{0}
\renewcommand{\theHequation}{D\arabic{equation}}
In this section, we derive the SDP for $S^{(1)}_{\gamma, \rm NS}$, the one-shot zero-error $\gamma$-cost shadow simulation cost assisted by no-signaling codes. Along the way, we derive SDPs for some related quantities, which are of interest on their own.

To begin with, the minimum error and minimum sampling cost of shadow simulation with a noiseless channel can be solved by SDPs in Lemma~\ref{lemma:noise_sim_min_err}.
We omit the proof here as it is very similar to the proof of Lemma~\ref{lemma:shadow_comm_min_err}.

\begin{lemma}\label{lemma:noise_sim_min_err}
    For an identity channel $\cI_d$ with dimension $d$ and a target quantum channel $\cM_{A'\to B'}$, the minimum error of the simulation from $\cI_d$ to $\cM_{A'\to B'}$ assisted by ${\rm NS}$ codes with a cost budget $\gamma$ is given by the following SDP:
    \begin{subequations}
    \begin{align}
        \varepsilon^*_{\gamma, \rm NS} (\cI_d,\cN) = \inf &\; \varepsilon \\
        \text{\rm s.t.} &\; Z_{A'B'} \geq 0,~ Z_{A'B'} \geq J^\cM_{A'B'} - Y^+_{A'B'} + Y^-_{A'B'}, \\
        &\; \tr_{B'}[Z_{A'B'}] \leq \frac{1}{2}\lrp{2\varepsilon + 1 - \frac{\tr\lrb{V^+_{B'} - V^-_{B'}}}{d^2}}\idop_{A'}, \\
        &\; 0 \leq Y^\pm_{A'B'} \leq \idop_{A'} \ox V^\pm_{B'},~ \tr_{B'}\lrb{Y^\pm_{A'B'}} = \frac{\tr\lrb{V^\pm_{B'}}}{d^2} \idop_{A'},~ \frac{\tr\lrb{V^+_{B'} + V^-_{B'}}}{d^2} \leq \gamma. \label{appeq:noise-sim-sdp_y-and-v-euqality}
    \end{align}
    \end{subequations}
    Similarly, the minimum sampling cost $\gamma_{\varepsilon, \rm NS}^*(\cI_d, \cM_{A'\to B'})$ with an error tolerance $\varepsilon$ is given by changing the optimization objective of the above SDP to $\tr\lrb{V^+_{B'} + V^-_{B'}} / d^2$ and removing the condition $\tr\lrb{V^+_{B'} + V^-_{B'}} / d^2 \leq \gamma$.
\end{lemma}

%%%%%%%%%%%%%%%%%%%%%%%%%%%%%%%%%%%%%%%%%
Provided with Lemma~\ref{lemma:noise_sim_min_err}, we now give an SDP for the zero-error minimum sampling cost of the shadow simulation of a noisy channel via a noiseless one.

\begin{lemma}\label{lemma:zero-error_simulation_noisy_channel}
    The zero-error minimum sampling cost of the shadow simulation of a channel $\cM_{A'\to B'}$ via a $d$-dimensional identity channel $\cI_d$ assisted by ${\rm NS}$ codes is given by the following SDP:
    \begin{subequations}\label{appeq:zero-error_simulation_noisy_channel_sdp}
    \begin{align}
        \gamma_{0, \rm NS}^*(\cI_d, \cM_{A'\to B'}) = \inf &\; \frac{\tr\lrb{V^+_{B'} + V^-_{B'}}}{d^2} \\
        \text{\rm s.t.} &\; \tr\lrb{V^+_{B'} - V^-_{B'}} = d^2,~ \tr_{B'}\lrb{R_{A'B'}} = \frac{\tr\lrb{V^-_{B'}}}{d^2} \idop_{A'}, \\
        &\; J^\cM_{A'B'} + R_{A'B'} \leq \idop_{A'} \ox V^+_{B'},~ 0 \leq R_{A'B'} \leq \idop_{A'} \ox V^-_{B'}
    \end{align}
    \end{subequations}
\end{lemma}
\begin{proof}
As in the shadow communication setting, the zero-error simulation of the channel $\cM$ requires $J^\cN_{A'B'} = Y^+_{A'B'} - Y^-_{A'B'}$ and $\tr\lrb{V^+_{B'} - V^-_{B'}} = d^2$ (see the proof of Lemma~\ref{lemma:communication_zero-error_sample-complexity}). The equalities $\tr_{B'}\lrb{Y^\pm_{A'B'}} = \frac{\tr\lrb{V^\pm_{B'}}}{d^2} \idop_{A'}$ in condition~\eqref{appeq:noise-sim-sdp_y-and-v-euqality} can be equivalently written as
\begin{align}
    \tr_{B'}\lrb{Y^+_{A'B'} + Y^-_{A'B'}} = \frac{\tr\lrb{V^+_{B'} + V^-_{B'}}}{d^2} \idop_{A'} \quad\text{and}\quad \tr_{B'}\lrb{Y^+_{A'B'} - Y^-_{A'B'}} = \frac{\tr\lrb{V^+_{B'} - V^-_{B'}}}{d^2} \idop_{A'}.
\end{align}
Note that the latter equality $\tr_{B'}\lrb{Y^+_{A'B'} - Y^-_{A'B'}} = \tr\lrb{V^+_{B'} - V^-_{B'}} \idop_{A'}/d^2$ can be removed because it is already implied by $J^\cN_{A'B'} = Y^+_{A'B'} - Y^-_{A'B'}$ and $\tr\lrb{V^+_{B'} - V^-_{B'}} = d^2$ with the observation that $\tr_{B'}\lrb{J^\cN_{A'B'}} = \idop_{A'}$ for $\cN$ being a quantum channel. Hence, we arrive at the following SDP:
\begin{subequations}
\begin{align}
    \gamma_{0, \rm NS}^*(\cI_d, \cM_{A'\to B'}) = \inf &\; \frac{\tr\lrb{V^+_{B'} + V^-_{B'}}}{d^2} \\
    \text{\rm s.t.} &\; J^\cM_{A'B'} = Y^+_{A'B'} - Y^-_{A'B'},~ \tr\lrb{V^+_{B'} - V^-_{B'}} = d^2, \\
    &\; 0 \leq Y^\pm_{A'B'} \leq \idop_{A'} \ox V^\pm_{B'},~ \tr_{B'}\lrb{Y^+_{A'B'} + Y^-_{A'B'}} = \frac{\tr\lrb{V^+_{B'} + V^-_{B'}}}{d^2} \idop_{A'}.
\end{align}
\end{subequations}
Denoting $R_{A'B'} \coloneqq Y^-_{A'B'}$, writing $Y^+_{A'B'}$ as $J^\cM_{A'B'} + R_{A'B'}$, and exploiting $\tr_{B'}\lrb{J^\cM_{A'B'}} = \idop_{A'}$, one can obtain the claimed SDP.
\end{proof}

%%%%%%%%%%%%%%%%%%%%%%%%%%%%%%%%%%%%%%%%%
From this lemma, we can arrive at the following SDP for the one-shot zero-error $\gamma$-cost shadow simulation cost.

\begin{theorem}
    The one-shot zero-error $\gamma$-cost simulation cost of a quantum channel $\cM_{A'\to B'}$ assisted by no-signaling codes is given by the following SDP:
    \begin{subequations}\label{sdp:zero-error_simulation_cost}
    \begin{align}
        S^{(1)}_{\gamma, \rm NS}(\cM) = \inf &\; \log_2~  \left\lceil \sqrt{\tr\lrb{V_{B'}}} \right\rceil \\
        \text{\rm s.t.} &\; J^\cM_{A'B'} + R_{A'B'} \leq \idop_{A'} \ox \frac{\gamma + 1}{2} V_{B'},~ 0 \leq R_{A'B'} \leq \idop_{A'} \ox W_{B'}, \\
        &\; \tr_{B'}\lrb{R_{A'B'}} = \frac{\gamma - 1}{2} \idop_{A'},~ \tr\lrb{W_{B'}} = \frac{\gamma - 1}{2} \tr\lrb{V_{B'}}.
    \end{align}
    \end{subequations}
\end{theorem}
\begin{proof}
    According to the SDP given in Lemma~\ref{lemma:zero-error_simulation_noisy_channel}, we can write the one-shot simulation cost $S^{(1)}_{\gamma, \rm NS}(\cM)$ as an optimization problem by substituting $d^2$ with $\tr\lrb{V^+_{B'} - V^-_{B'}}$:
    \begin{subequations}
    \begin{align}
        S^{(1)}_{\gamma, \rm NS}(\cM) = \inf &\; \log_2~  \left\lceil \sqrt{\tr\lrb{V^+_{B'} - V^-_{B'}}} \right\rceil \\
        \text{\rm s.t.} &\; \tr_{B'}\lrb{R_{A'B'}} = \frac{\tr\lrb{V^-_{B'}}}{\tr\lrb{V^+_{B'} - V^-_{B'}}} \idop_{A'},~ \frac{\tr\lrb{V^+_{B'} + V^-_{B'}}}{\tr\lrb{V^+_{B'} - V^-_{B'}}} \leq \gamma, \label{appeq:sim-cost-proof-sdp_one}\\
        &\; J^\cM_{A'B'} + R_{A'B'} \leq \idop_{A'} \ox V^+_{B'},~ 0 \leq R_{A'B'} \leq \idop_{A'} \ox V^-_{B'}.
    \end{align}
    \end{subequations}
    Note that the inequality in condition~\eqref{appeq:sim-cost-proof-sdp_one} can be restricted to equality while keeping the optimized value unchanged. This is true because if $R_{A'B'}$ and $\bar{V}^\pm_{B'}$ is a set of optimal solution with $\tr\lrb{V^+_{B'} + V^-_{B'}} / \tr\lrb{\bar{V}^+_{B'} - \bar{V}^-_{B'}} = \bar{\gamma} < \gamma$, then $R'_{A'B'} \coloneqq R_{A'B'} + (\gamma - \bar{\gamma}) \idop_{A'B'} / 2d_{B'}$ and $\bar{V}'^\pm_{B'} \coloneqq \bar{V}^\pm_{B'} + (\gamma - \bar{\gamma}) \tr\lrb{\bar{V}^+_{B'} - \bar{V}^-_{B'}} \idop_{B'} / 2d_{B'}$ also form a set of optimal solution such that $\tr\lrb{\bar{V}'^+_{B'} + \bar{V}'^-_{B'}} / \tr\lrb{\bar{V}'^+_{B'} - \bar{V}'^-_{B'}} = \gamma$.
    Note that for $\bar{R}'^\pm_{A'B'}$ and $\bar{V}'^\pm_{B'}$ to be valid solution, it must be true that $\tr\lrb{\bar{V}^+_{B'} - \bar{V}^-_{B'}} \geq 1$ so that $\idop_{A'} \ox \bar{V}'^-_{B'} \geq \bar{R}'_{A'B'}$. This is indeed the case because from the constraint $R_{A'B'} \leq \idop_{A'} \ox V^-_{B'}$ we have $\tr_{B'}\lrb{\bar{R}_{A'B'}} \leq \tr\lrb{\bar{V}^-_{B'}} \idop_{A'}$. For the equality $\tr_{B'}\lrb{\bar{R}_{A'B'}} = \tr\lrb{\bar{V}^-_{B'}} \idop_{A'} / \tr\lrb{\bar{V}^+_{B'} - \bar{V}^-_{B'}}$ to hold, $\tr\lrb{\bar{V}^+_{B'} - \bar{V}^-_{B'}}$ has to be larger than or equal to $1$.
    
    By changing the inequality in condition~\eqref{appeq:sim-cost-proof-sdp_one} to $\tr\lrb{V^+_{B'} + V^-_{B'}} / \tr\lrb{V^+_{B'} - V^-_{B'}} = \gamma$, it follows that $\tr[V^-_{B'}] = (\gamma-1)\tr[V^+_{B'}] / (\gamma + 1)$ and thus $\tr_{B'}\lrb{R_{A'B'}} = (\gamma - 1) \idop_{A'} / 2$. Changing the variables to $V_{B'} \coloneqq 2V^+_{B'} / (\gamma + 1)$ and $W_{B'} \coloneqq V^-_{B'}$ gives the claimed SDP.
\end{proof}

Similar to the one-shot zero-error shadow capacity, the SDP above implies that $S^{(1)}_{\gamma,\rm NS}$ is a generalization of the no-signaling-assisted one-shot zero-error quantum simulation cost $S^{(1)}_{\rm NS}$ as it reduces to the SDP for the latter~\cite{duan2015no, fang2019quantum} when $\gamma = 1$.
To showcase the difference between $S^{(1)}_{\gamma,\rm NS}$ and $S^{(1)}_{\rm NS}$, we consider two independent uses of amplitude damping channel $\cN_{\rm AD}^{\ox 2}$, dephasing channel $\cN_{\rm deph}^{\ox 2}$, and depolarizing channel $\cN_{\rm depo}^{\ox 2}$.

\begin{figure}[h]
    \centering
    \includegraphics[width=0.49\textwidth]{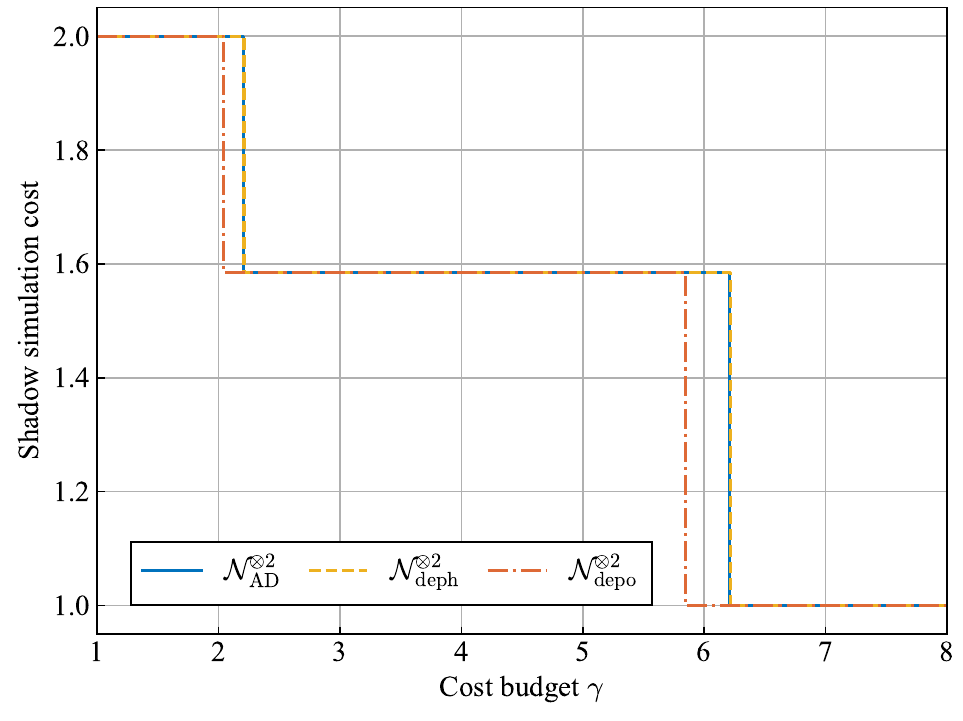}
    \caption{{\bf Comparison between the conventional no-signaling-assisted one-shot zero-error quantum simulation cost $S^{(1)}_{\rm NS}$ and no-signaling-assisted one-shot zero-error shadow simulation cost $S^{(1)}_{\gamma,\rm NS}$ under different cost budgets $\gamma$.} Compared with the quantum case, where the simulation cost is $2$ for all channels, lower simulation cost is achieved for each channel with increased cost budget.}
    \label{fig:sim_cost}
\end{figure}

As in the task of shadow communication, we present numerical results on these channels at a low noise level ($p=0.9$) in Fig.~\ref{fig:sim_cost}.
The zero-error simulation costs of these channels decrease from $2$ (quantum simulation cost) to $1$ with increased cost budget. Again, the stepwise changes in the zero-error simulation cost show that we can attain computational power beyond purely quantum protocols by trading in more computational resources.

%%%%%%%%%%%%%%%%%%%%%%%%%%%%%%%%%%%%%%%%%
In the main text, we presented the minimum sampling cost of simulating a high-dimensional identity channel with a low-dimensional one. Now, we prove this result.
\renewcommand\theproposition{\ref{thm:high_id_sim_cost}}
\setcounter{proposition}{\arabic{proposition}-1}
\begin{theorem}
    Given identity channels $\cI_d$ and $\cI_{d'}$ with $d' \geq d \geq 2$, the minimum sampling cost of an exact shadow simulation of $\cI_{d'}$ using $\cI_d$ and no-signaling resources is
    \begin{align}
        \gamma^*_{0, \rm NS} \lrp{\cI_d, \cI_{d'}} = 2 \lrp{\frac{d^{\prime}}{d}}^2 - 1.
    \end{align}
\end{theorem}
\renewcommand{\theproposition}{\arabic{proposition}}
\begin{proof}
The Choi operators of the noiseless channel $\cI_d$ from system $A$ to system $B$ and the noiseless channel $\cI_{d'}$ from system $A'$ to system $B'$ are $J^{\cI_d}_{AB} = \Gamma_{AB}$ and $J^{\cI_{d'}}_{A'B'} = \Gamma_{A'B'}$, respectively.
It is straightforward to verify that
\begin{align}
    \bar{V}^+_{B'} = d' \idop_{B'},\quad \bar{V}^-_{B'} = \frac{d'^2 - d^2}{d'} \idop_{B'},\quad \bar{R}_{A'B'} = \frac{d'(d'^2-d^2)}{(d'^2-1)d^2} \idop_{A'B'} - \frac{d'^2-d^2}{(d'^2-1)d^2} \Gamma_{A'B'}
\end{align}
form a feasible solution to the SDP for $\gamma_{0, \rm NS}^*(\cI_d, \cI_d')$ as given in Lemma~\ref{lemma:zero-error_simulation_noisy_channel}, implying
\begin{align}\label{appeq:high_dim_id_samp_complex_upper}
    \gamma_{0, \rm NS}^*(\cI_d, \cI_d') \leq \frac{1}{d^2} \tr\lrb{\bar{V}^+_{B'} + \bar{V}^-_{B'}} = \frac{2d'^2}{d^2} - 1.
\end{align}

Using the Lagrange dual function, we can derive that the following problem is the dual problem associated with SDP~\eqref{appeq:zero-error_simulation_noisy_channel_sdp}:
\begin{subequations}
\begin{align}
    \gamma_{0, \rm NS}^*(\cI_d, \cM_{A'\to B'}) \geq \sup &\; \tr\lrb{M_{A'B'} J^\cM_{A'B'}} - \lambda \\
    \text{\rm s.t.} &\; M_{A'B'} \geq 0,~ N_{A'B'} \geq 0,~ M_{A'B'} + N_{A'B'} + K_{A'} \ox \idop_{B'} \geq 0, \\
    &\; d^2 \tr_{A'}\lrb{M_{A'B'}} = (1+\lambda) \idop_{B'},~ d^2 \tr_{A'}\lrb{N_{A'B'}} = (1-\lambda - \tr\lrb{K_{A'}}) \idop_{B'}.
\end{align}
\end{subequations}
It is straightforward to verify that
\begin{align}
    \bar{\lambda} = 1,\quad \bar{M}_{A'B'} = \frac{2}{d^2} \Gamma_{A'B'},\quad \bar{N}_{A'B'} = 0,\quad \bar{K}_{A'} = 0.
\end{align}
form a feasible solution to the dual problem for $\gamma_{0, \rm NS}^*(\cI_d, \cI_{d'})$, implying
\begin{align}\label{appeq:high_dim_id_samp_complex_lower}
    \gamma_{0, \rm NS}^*(\cI_d, \cI_{d'}) \geq \tr\lrb{\bar{M}_{A'B'} J^{\cI_{d'}}_{A'B'}} - \bar{\lambda} = \frac{2d'^2}{d^2} - 1.
\end{align}
Combining Eqs.~\eqref{appeq:high_dim_id_samp_complex_upper} and~\eqref{appeq:high_dim_id_samp_complex_lower}, we conclude that
\begin{align}
    \gamma_{0, \rm NS}^*(\cI_d, \cI_d') = \frac{2d'^2}{d^2} - 1,
\end{align}
which completes the proof.
\end{proof}

\end{document}